\documentclass[12pt,onecolumn,journal]{IEEEtran}
\normalsize
\usepackage{amsmath}
\usepackage{mathtools}
\usepackage{amsfonts}
\usepackage{amssymb}
\usepackage{amsthm}
\usepackage{calc}
\usepackage{color}
\usepackage{epsfig}
\usepackage{setspace}
\usepackage{pstricks}
\usepackage{cancel}
\usepackage{multirow}
\usepackage{mathrsfs}
\usepackage{algorithm}
\usepackage{algorithmic}
\usepackage{graphicx}
\usepackage{subcaption}
 \usepackage{enumitem}
\usepackage{cite}

\usepackage{pgfplots} 
\pgfplotsset{compat=newest} 
\pgfplotsset{plot coordinates/math parser=false}

\newtheorem{defn}{Definition}
\newtheorem{thm}{{\cal T}heorem}
\newtheorem{cor}{Corollary}
\newtheorem{prop}{Proposition}
\newtheorem{lem}{Lemma}
\newtheorem{conj}{Conjecture}
\newtheorem{condition}{Condition}
\newtheorem{constr}{Construction}
\newtheorem{note}{Remark}
\newtheorem{claim}{Claim}
\newtheorem{example}{Example}
\newcommand{\bit}{\begin{itemize}}
\newcommand{\eit}{\end{itemize}}
\newcommand{\bcor}{\begin{cor}}
\newcommand{\ecor}{\end{cor}}
\newcommand{\beq}{\begin{equation}}
\newcommand{\eeq}{\end{equation}}
\newcommand{\beqn}{\begin{equation}}
\newcommand{\eeqn}{\end{equation}}
\newcommand{\bea}{\begin{eqnarray}}
\newcommand{\eea}{\end{eqnarray}}
\newcommand{\bean}{\begin{eqnarray*}}
\newcommand{\eean}{\end{eqnarray*}}
\newcommand{\ben}{\begin{enumerate}}
\newcommand{\een}{\end{enumerate}}
\newcommand{\bdefn}{\begin{defn}}
\newcommand{\edefn}{\end{defn}}
\newcommand{\bnote}{\begin{note}}
\newcommand{\enote}{\end{note}}
\newcommand{\bprop}{\begin{prop}}
\newcommand{\eprop}{\end{prop}}
\newcommand{\blem}{\begin{lem}}
\newcommand{\elem}{\end{lem}}
\newcommand{\bthm}{\begin{thm}}
\newcommand{\ethm}{\end{thm}}
\newcommand{\bconj}{\begin{conj}}
\newcommand{\econj}{\end{conj}}
\newcommand{\bconstr}{\begin{constr}}
\newcommand{\econstr}{\end{constr}}
\newcommand{\bpf}{\begin{proof}}
\newcommand{\epf}{\end{proof}}

\begin{document}
	\title{Straggler Mitigation with Tiered Gradient Codes}
\author{Shanuja Sasi, V. Lalitha, Vaneet Aggarwal,~\IEEEmembership{Senior Member,~IEEE, } and B. Sundar Rajan,~\IEEEmembership{Fellow,~IEEE} \thanks{S. Sasi and B. Sundar Rajan are with the Department of Electrical Communication Engineering at Indian Institute of Science Bangalore 560012, email: $\{$shanuja,bsrajan$\}$@iisc.ac.in. S. Sasi was with Purdue University when this work was performed. V. Lalitha is with the Signal Processing and Communications Research Center, IIIT Hyderabad, email: lalitha.v@iiit.ac.in. V.  Aggarwal is with the School of Industrial Engineering and the School of Electrical and Computer Engineering at Purdue University, West Lafayette, IN 47907, email: vaneet@purdue.edu.}}


	\maketitle
	\begin{abstract}
	
	Coding theoretic techniques have been proposed for  synchronous Gradient Descent (GD) on multiple servers to mitigate stragglers. These techniques provide the flexibility that the job is complete when any $k$ out of $n$ servers finish their assigned tasks. The task size on each server is found based on the values of $k$ and $n$. However, it is assumed that all the $n$ jobs are started when the job is requested. In contrast, we assume a tiered system, where we start with $n_1\ge k$ tasks, and on completion of $c$ tasks, we start $n_2-n_1$ more tasks. The aim is that as long as $k$ servers can execute their tasks, the job gets completed. This paper exploits the flexibility that not all servers are started at the request time to obtain the achievable task sizes on each server. The task sizes are in general lower than starting all $n_2$ tasks at the request times thus helping achieve lower task sizes which helps to reduce both the job completion time and the total server utilization. 
\end{abstract}

\section{Introduction}


Many distributed machine learning applications require multiple servers to perform distributed computation of gradient descent. Distributed gradient descent involves division of gradient descent tasks across multiple servers, and the job is finished when all the tasks are complete. The slowest tasks that determine the job execution time are called stragglers. Coding theoretic techniques have been proposed to achieve high-quality algorithmic results in the face of uncertainty, including mitigation of stragglers. \cite{tandon2017gradient,ye2018communication,lee2018speeding,dutta2016short,li2018fundamental,wan2018fundamental,yu2018straggler}.
These approaches have been shown to be essential to manage stragglers in distributed computation of gradient descent \cite{tandon2017gradient,ye2018communication,dutta2018slow}. However, these approaches assume that all the distributed tasks are started at the same time, which can be shown to have a large server utilization cost. To alleviate that, this paper aims to provide a tiered framework for efficient gradient code designs that allow for starting certain tasks at the completion of some tasks with an aim to have an efficient tradeoff between the completion time of the job and the server utilization cost to complete the tasks.

In this paper, we propose a coding-theoretic approach for gradient coding, called \textit{Tiered Gradient Coding}. Initially at the service request time, tasks are launched on $n_1$ servers. On the completion of tasks from $c$ of the servers, tasks are launched on $n_2-n_1>0$ more servers, where $n_2$ is the total number of servers. We note that the earlier studied gradient codes \cite{tandon2017gradient,ye2018communication,dutta2018slow} do not have two tiers and the tasks for $n_2$ servers are decided at the same time. By having the flexibility of obtaining the results from $c$ servers leads to reduction of per-server workload as compared to deciding tasks for $n_2$ servers at the same time. Consider as an example of gradient coding scheme in Fig. \ref{fig: n2.}, where the data is split into $4$ partitions $D_1,D_2,D_3$ and $D_4$. Server $W_1$ computes the gradients $g_1,g_2$ and $g_3$ of the  partitions $D_1,D_2$ and $D_3$ respectively. Similarly, server $W_2$ computes the gradients $g_2,g_3$ and $g_4$, server $W_3$ computes the gradients $g_3,g_4$ and $g_1$ and server $W_4$ computes the gradients $g_4,g_1$ and $g_2$. Each server sends a linear combination of the gradients they have computed. It is enough to get the results from any two servers to calculate the overall sum of gradients. The techniques to calculate the linear combination are provided in \cite{tandon2017gradient}. The computation cost per server is proportional to $\frac{3}{4}$.  In  Fig. \ref{fig: tier}, we describe the proposed tiered gradient coding framework, where the data is split into $3$ partitions $D'_1,D'_2$ and $D'_3$. Initially, only three servers ($W_1,W_2$ and $W_3$) are launched. Server $W_1$ computes the gradients $g'_1$ and $g'_2$ of the  partitions $D'_1$ and $D'_2$ respectively. Similarly, server $W_2$ computes the gradients $g'_2$ and $g'_3$ and server $W_3$ computes the gradients $g'_3$ and $g'_1$. Without loss of generality, assume that the server $W_1$ finishes its task first, {\em i.e.}, $W_1$ sends a linear combination of the gradients $g'_1$ and $g'_2$. Then, server $W_4$ is launched, which computes the gradients $g'_3$ and $g'_1$ of the  partitions $D'_3$ and $D'_1$ respectively. The partitions assigned to $W_4$ depends on the server which had completed the task initially. Master waits for one of the servers - $W_2,W_3$ and $W_4$ to complete the task. Master can calculate the sum from the result from $W_1$ and any one of the servers - $W_2,W_3$ or $W_4$. For example, if $W_4$ completes the task first, $g'_1+g'_2+g'_3= (\frac{g'_1}{2}+g'_2) + (\frac{g'_1}{2}+g'_3)$. Thus, we see that the task per server reduces from $\frac{3}{4}$ to $\frac{2}{3}$ for the same number of four servers, and both schemes guarantee that any two servers completion can provide the required computation result. 

\begin{figure*}[!t]
	\begin{minipage}[]{.4\textwidth}
		\centering
		\includegraphics[width=\textwidth]{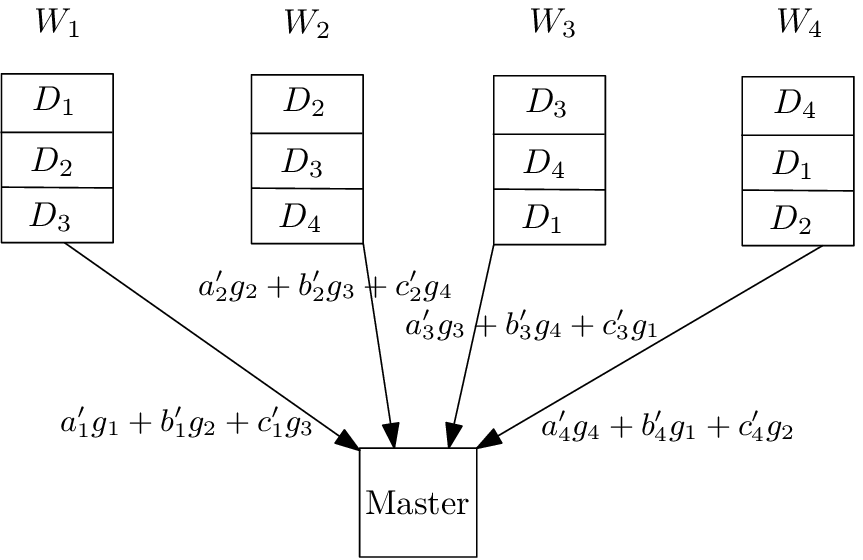}
		\caption{Gradient coding: Each server transmits a scalar and master calculates the sum from the result of any two servers, with the total number of servers being four.}
		\label{fig: n2.}
	\end{minipage}
	\hspace{.2in}
	\begin{minipage}[]{.55\textwidth}
		\centering
		\includegraphics[width=\textwidth]{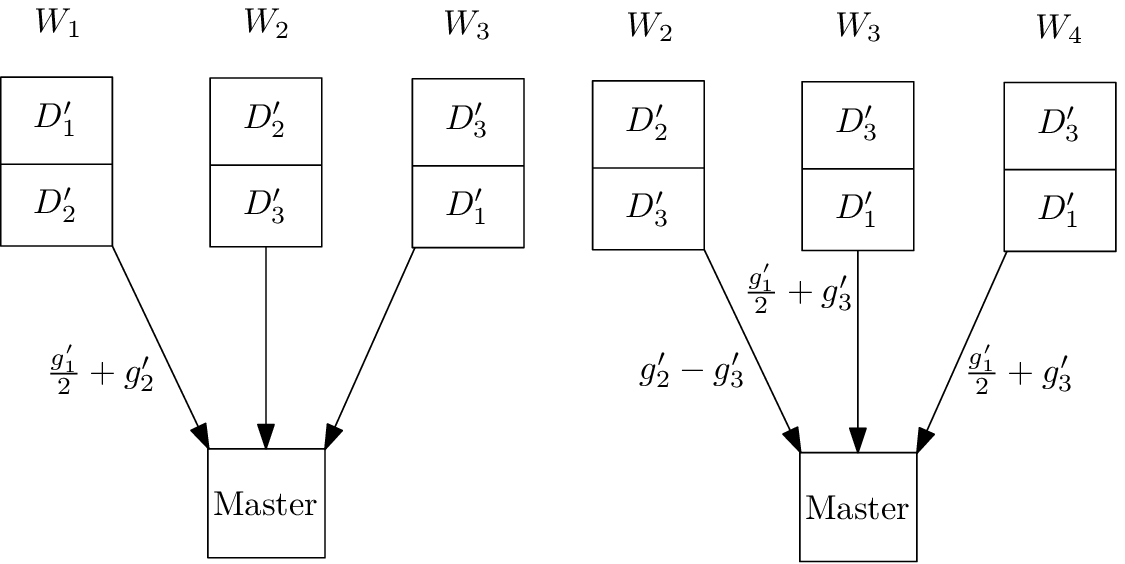}
		\caption{Tiered gradient coding: Initially $3$ servers, i.e. $W_1,W_2$ and $W_3$, are launched. Without loss of generality, assume $W_1$ completes the task first. Then the foruth server $W_4$ is launched. Master waits for one more server to finish the task to calculate the sum.}
		\label{fig: tier}
	\end{minipage}
\end{figure*}

\begin{figure*}[hbt!]
	\begin{subfigure}[]{.31\textwidth}
		\resizebox{\textwidth}{!}{%
			\includegraphics[width=\linewidth]{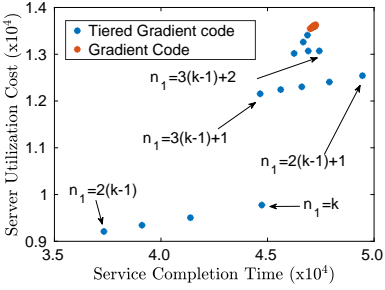}
		}
		\caption{Task completion time distributed as SE1}
		\label{fig:SUC_Vs_SCT_k5_fixed_p}
	\end{subfigure}
	\hspace{.1in}
	\begin{subfigure}[]{.31\textwidth}
		\resizebox{\textwidth}{!}{%
			\includegraphics[width=\linewidth]{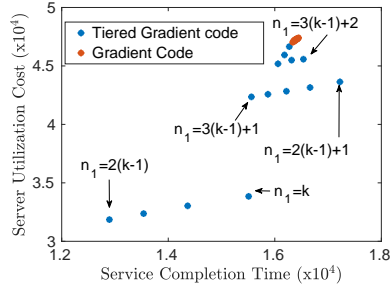}
		}
		\caption{Task completion time distributed as SE2}
		\label{fig:SUC_Vs_SCT_k5_var_p}
	\end{subfigure}
	\hspace{.1in}
	\begin{subfigure}[]{.31\textwidth}
		\resizebox{\textwidth}{!}{%
			\includegraphics[width=\linewidth]{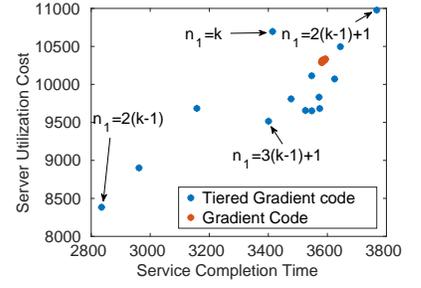}
		}
		\caption{Task completion time distributed as Pa}
		\label{fig:SUC_Vs_SCT_k5_pareto}
	\end{subfigure}
	\caption{Server Utilization Cost as a function of Service Completion Time when we vary $n_1 \in [k,n_2]$ for $n_2=15,c=1$ and $k=5$.}\label{Figcase1}
\end{figure*}

We note that waiting to launch tasks $n_2-n_1$ servers can affect the job completion time negatively, while the decreased task size per server affects the job completion time positively, thus making it apriori unclear whether the completion time increases or decreases. Similarly, server utilization cost (total time any server is being used for computation) may also increase if the completion time is increased, and decrease since $n_2-n_1$ servers are not used till $c$ complete the computation. A tradeoff between the two metrics of completion time and server utilization cost have been considered for coded tasks \cite{aktas2018straggler}, while we show that efficient coding can also decrease the task size when such multi-forking capability (starting $n_2$ after $c$ have finished) can lead to decreased task size further reducing both the metrics. As an example, we consider $n_2=15$, $c=1$, and $k=5$, where $k$ is the number of servers that must complete execution for the job to be completed. For three models of the time taken for each server to complete the task (with the models described in Section \ref{sec:nume}), Fig. \ref{Figcase1} shows for different values of $n_1$, the tradeoff (between service completion time and server utilization cost) points that can be achieved with tiered-gradient codes as compared to gradient coding where all $n_2$ tasks must be decided apriori. Since the gradient codes are independent of $n_1$, they are only a point, while we achieve different tradeoff points for different $n_1$. We see that the proposed codes perform significantly better in both the completion time and server utilization costs and efficient parameters can be decided based on the application requirements. 

The key contribution of the paper is a new framework for tiered gradient codes which allows for a delayed start of the tasks at the servers. A novel code construction is provided that exploits this flexibility, and reduces the amount of computation that each server has to perform.

{\bf Notation: } Throughout this paper, we let $d$ denote the number of samples, $n_2$ denote the total number of servers, $Q$ denote the number of data partitions, $s$ denote the number of stragglers/failures. Let $k$ denote the minimum number of servers required to finish their task such that the overall gradient can be computed. Let $n_1<n_2$ denote the number of servers launched in the first phase. We wait for $c<k$ servers to finish their job first when $n_1$ servers are launched. In the second phase the rest of $n_2-n_1$ servers are launched. The $n_2$ servers are denoted as $\{W_1,W_2, \ldots , W_{n_2}\}$. The partial gradients over $Q$ data partitions are denoted as $\{g_1,g_2,...,g_Q\}$. All matrices under consideration are over real numbers. Let $[z]$ denote the set $\{1,2,...,z\}$ and $[z_1,z_2]$ denote the set $\{z_1,z_1+1,...,z_2\}.$

\section{Related Work}
\label{related work}
Coding-theoretic techniques to mitigate the effect of stragglers in gradient computation were introduced in \cite{tandon2017gradient}. In \cite{ye2018communication}, coding techniques to reduce the running time of distributed learning tasks have been provided. A stochastic block code and an efficient decoding method for approximate
gradient recovery are provided in  \cite{charles2018gradient}. A distributed computing scheme called Batched Coupon’s Collector to mitigate the effect of stragglers in gradient methods is proposed in \cite{li2018near}. 
In \cite{halbawi2018improving}, a straggler mitigation scheme that facilitates the implementation of distributed gradient descent
in a computing cluster is presented. They also  proposed a theoretical delay model which allows to minimize the expected running time. In \cite{raviv2017gradient}, an approximate variant of the gradient coding problem is introduced, in which approximate gradient computation is done instead of the exact computation.

A cost vs. latency analysis of using simple replication or erasure coding for straggler mitigation in executing jobs with many tasks is studied in \cite{wang2015efficient,aktas2018straggler}. Both in \cite{wang2015efficient} and \cite{aktas2018straggler}, the authors have showed that the delayed relaunch of stragglers yields significant reduction in cost and latency. In this paper, we show that efficient coding further allows reduction of task size per server with a delayed execution of tasks, a flexibility which had not been studied earlier. Thus, our coding-theoretic techniques can further help reduce the job completion time by exploiting a better choice of the parameters since starting more servers need larger task size at each server in general. Thus, this paper aims at finding efficient code constructions that minimizes per-server task sizes with the flexibility of tiered launching of tasks.

\section{Review of Gradient Codes}\label{review_gradient}
\label{tiered grad framework}
\subsection{Distributed Gradient Descent Computation}
Given a dataset $D$ with $d$ examples, $D=\{(X_i,Y_i)\}_{i=1}^{d}$, where $X_i \in \mathbb{R}^{p}$ and $Y_i \in \mathbb{R}$, we want to learn parameters $\beta \in  \mathbb{R}^{p}$ by minimizing a generic loss function $L(D;\beta) = \sum_{i=1}^{d} L(X_i,Y_i;\beta)$. We update the parameter $\beta$ according to the following rule:
$\beta^{(t+1)} = h(\beta^{(t)}; g^{(t)}), $ where $g^{(t)} =  \nabla L(D; \beta^{(t)}) = \sum_{i=1}^{d} \nabla L(X_i,Y_i; \beta^{(t)})$ is the gradient of the loss at the current estimate of the parameters and $h$ is a gradient-based optimizer. We consider the problem of distributed synchronized gradient descent where the $d$ data samples are divided into $Q$ partitions, $D_1, D_2, \ldots, D_Q$. The partial gradient computed on the $j^{\text{th}}$ partition is given by $
g_j^{(t)} = \sum_{(X,Y) \in D_j} \nabla L(X,Y; \beta^{(t)}). 
$
The overall gradient required to compute the update on $\beta^{(t)}$ is given by $g^{(t)}=\sum_{j=1}^Q g_j^{(t)}$. We will omit the superscript $t$ in this paper to simplify the notation. Next, we provide a review of two classes of conventional gradient codes known as fractional repetition gradient codes and cyclic repetition gradient codes \cite{tandon2017gradient}.   

\subsection{Gradient Coding Framework}
For $n_2$ workers and $Q$ data partitions, we have a set of linear equations:
$AF = 1_{f \times Q}, $
where $f$ denotes the number of combinations of surviving servers/non-stragglers, $1_{f \times Q}$ is the all 1
matrix of dimension $f \times Q$ and we have matrices $A \in \mathbb{R}^{f \times n_2}$, $F \in \mathbb{R}^{n_2 \times Q}$.
The $i^{th}$ row of $F, \bold{f_i},$ is associated with the $i^{th}$ server $W_i$. The support of $\bold{f_i}, supp(\bold{f_i})$, represents the data partitions corresponding to the server $W_i$ and the entries of $\bold{f_i}$ encode a linear combination
over their gradients that server $W_i$ transmits. Let $\bold{g} \in \mathbb{R}^{Q \times p}$ be a matrix with each row being the  partial gradient of a data partition i.e.
$\bold{g} = [g_1, g_2, ..., g_Q] $.
Then, server $W_i$ transmits $\bold{f_i} \bold{g}$. Each row of $A$, denoted by $\bold{a_i}$, is associated with a specific straggler scenario, to which tolerance is desired. In particular, any row $\bold{a_i}$, with support  $supp(\bold{a_i})$, corresponds to the scenario where the server indices in $supp(\bold{a_i})$ are non-stragglers. The entries of $\bold{a_i}$ encode a linear combination
which, when taken over the transmitted gradients of the non-straggler servers, $\{\bold{f_u}\bold{g}\}_{ u \in supp(\bold{a_i})}$,
would yield the full gradient. We refer to this system as $(n_2,k)$ gradient code where $k$ is the number of non stragglers.

\subsection{Fractional Repetition Gradient Codes \cite{tandon2017gradient}} Consider the case when $n_2-k+1$ divides $ n_2$. Let $Q = n_2$. Consider the following matrix $
F_j = 1_{(n_2-k+1) \times (n_2-k+1)}, \ \ 1 \leq j \leq \frac{n_2}{n_2-k+1}.
$
The matrix $F$ of the fractional repetition gradient code is constructed as follows: $
F = \begin{bmatrix} F_1& 0 & \ldots & 0 \\ 0 & F_2 & \ldots & 0 \\ \vdots & \vdots & \vdots & \vdots \\ 0 & 0 & \ldots & F_{\frac{n_2}{n_2-k+1}}  \end{bmatrix}.
$

\subsection{Cyclic Repetition Gradient Codes \cite{tandon2017gradient}} \label{sec:rev_cyclic}

This class of gradient codes exist for all values of $k$ and $n_2$. Let $Q = n_2$ and let the columns of $F$ be indexed by $[0, n-1]$. The support structure of the matrix $F$ is as follows: \begin{equation} \label{eq:suppf}
supp(\bold{f_i}) = [i-1,  i+(n-k-1)] \mod n_2.
\end{equation}
Now, we will present a randomized construction of the matrix $F$. Consider a matrix $H$ of size $(n_2-k) \times n_2$ whose first $(n_2-1)$ columns are picked at random i.i.d. from a Gaussian distribution $\mathcal{N}(0,1)$. The last column of $H$ is obtained as follows: $
H(:,n_2-1) = - \sum_{i=0}^{n_2-2}  H(:,i).$
Each vector $\bold{f_i}$ is calculated by solving the following equation $
\bold{f_i}|_{L_i } H(:,L_i)^T = 0, $
where $L_i$ is the  support of $\bold{f_i}$ as given by Equation \eqref{eq:suppf}.

The span condition for the conventional gradient code framework in \cite{tandon2017gradient} is given below.
\begin{lem} \cite{tandon2017gradient} \label{lem:span_b}
	Consider a gradient code $(F, A)$. For every $I \in [n_2]$,  such that $|I| = k$, we have
	$
	1_{1 \times Q} \in \text{span}\{\bold{f_i} | i \in I \}.
	$
\end{lem}

We will present the support condition which is a sufficient condition (\cite{tandon2017gradient}) to show that the randomized construction of $F$ satisfies the span condition in Lemma \ref{lem:span_b} with probability $1$. 

\begin{condition} \label{supp_cond_gradient}
	For every $\mathcal{T}$, which is a subset of $[n_2]$ of size $\ell$ ($1 \leq \ell \leq k$), we have
	$
	|\cup_{i \in \mathcal{T}} L_i| \geq (n_2-k) + \ell.
	$
\end{condition}
\section{Tiered Gradient Code Framework}
\label{tiered grad framework}

\subsection{Tiered Gradient Codes} 

In the conventional gradient code framework, we assume that there are $n_2$ servers which start computing the partial gradients assigned to them. We want to be able to compute the overall gradient whenever $k < n_2$ servers finish.
Each server sends a linear combination of the partial gradients which it has computed and sends it back to the master node.
The master node aggregates all the linear combinations of the partial gradients and performs a linear combination in turn to obtain the overall gradient $\sum_{j=1}^Q g_i$. 

In the tiered gradient code framework, we consider two phases. In the first phase, $n_1$ servers start computing the partial gradients from the data partitions assigned to them and $c$ out of the $n_1$ servers complete their gradient computation by the end of the first phase. In the second phase, $n_1-c$ servers continue their tasks which were started in the first phase and $n_2 - n_1$ servers start computing the partial gradients assigned to them in the second phase. The assignment of the data partitions to $n_2-n_1$ servers in the second phase is decided based on which $c$ servers out of the $n_1$ servers have finished. We want to be able to compute the overall gradient whenever $k$ servers out of the $n_2$ servers finish. This condition is equivalent to saying that we would need the results from any $k-c$ out of the $n_2-c$ servers to complete in the second phase, so that we can compute the overall gradient. We call this set up as $(n_1,n_2,k,c)$ tiered gradient coding. 


\subsection{Span Condition of Tiered Gradient Codes}\label{sec:span}

Consider the $Q$ partial gradients arranged in a column vector as $\bold{g} = [g_1, g_2, \ldots, g_Q]^T$.
Let $F$ denote a matrix of size $n_1 \times Q$ over $\mathbb{R}$. The $i^{\text{th}}$ row of the $F$ matrix is denoted by $\bold{f_i}$, $1 \leq i \leq n_1$. The support of $\bold{f_i}$ indicate the partial gradients which are computed on the $i^{\text{th}}$ server. The quantity $\bold{f_i} \bold{g}$ is the linear combination sent by the $i^{\text{th}}$ server to the master node. 

Let $M \subset [n_1]$ denote the set of $c$ servers which have finished their tasks at the end of the first phase and $\mathcal{M}$ denote the set of all possible $c$ subsets of $[n_1]$. Let $\{ B_M, M \in \mathcal{M}\}$ denote a set of matrices, each of size $(n_2 - n_1) \times Q$. The $i^{\text{th}}$ row of $B_M$ is denoted by $\bold{b_i}$, $1 \leq i \leq n_2-n_1$. The support of $\bold{b_i}$ indicate the partial gradients which are computed on the $i^{\text{th}}$ server among the $n_2-n_1$ servers started in the second phase. The quantity $\bold{b_i} \bold{g}$ is the linear combination sent by the $i^{\text{th}}$ server to the master node. 

Let $\{A_M, M \in \mathcal{M}\}$ denote a set of matrices, each of size $N \times n_2$. Columns of the $A_M$ matrix are indexed by the servers. The rows of the  $A_M$ matrix are denoted by $\bold{a}_i$. Each $\bold{a}_i$ has non-zeros in the $c$ positions corresponding to the subset $M$. The rows of the $A_M$ matrix are such that each row will have non-zeros in a distinct subset of $k-c$ out of the $n_2-c$ positions. Hence, the number of possible straggler configurations which can be tolerated by a tiered gradient code described above, for a given set of $c$ servers, is $N =  {n_2 - c \choose n_2 - k + c}$.

The condition for computing the overall gradient from the partial gradients in the tiered gradient code setup is given by
\begin{equation} \label{eq:cond_ab}
A_M \begin{bmatrix} F \\ B_M \end{bmatrix} = 1_{N \times Q}, \ \  \forall M \in \mathcal{M},
\end{equation}
where $1_{N \times Q}$ denotes a matrix all of whose entries are $1$.
We will refer to a $(n_1,n_2,k,c)$ tiered gradient code by $(F, \{(A_M, B_M), M \in \mathcal{M}\})$. Lemma \ref{lem:two_span_b} provides the necessary condition for a code to be a tiered gradient code.

\begin{claim}
	The partial gradients which are computed on the first $n_1$ servers have to constitute a $(n_1, k)$ gradient code.
\end{claim}

\begin{proof}
	We need to be able to compute the overall gradient whenever $k$ servers finish. This includes the $c$ servers which have computed the gradient in the first phase. Now, since $n_1 \geq k$, all the $k$ servers can be potentially from the first $n_1$ servers and since the property has to held for all possibilities of $c$ servers, the claim follows.
\end{proof}

\begin{lem}[Span Condition] \label{lem:two_span_b}
	Consider a tiered gradient code $(F, \{(A_M, B_M), M \in \mathcal{M}\})$. For every $M \in \mathcal{M}$,  $I_1 \subseteq [n_1] \setminus M$ and $I_2 \subseteq [n_2 - n_1]$, such that $|I_1 \cup I_2| = k-c$, it holds that
	$
	1_{1 \times Q} \in \text{span}\{\bold{f_i} | i \in M \cup I_1; \bold{b_i} | i \in I_2\}.
	$
\end{lem}

\begin{proof}
	
	We consider one row of Equation \eqref{eq:cond_ab} for a fixed $M$ and is given by
	\begin{equation} \label{eq:span_b}
	\bold{a_i} \begin{bmatrix} F \\ B_M \end{bmatrix} = 1_{1 \times Q}.
	\end{equation} 
	Let $L_i = supp(\bold{a_i}) = M \cup I_1 \cup \{ n_1 + 2\}$. The above equation can be rewritten as
	$
	\bold{a_i}|_{L_i} \begin{bmatrix} F|_{M \cup I_1} \\ B_M|_{I_2} \end{bmatrix} = 1_{1 \times Q}.
	$
	From Equation \eqref{eq:span_b}, it is clear that there exists a non-zero vector $\bold{a_i}|_{L_i}$ such that the above equation is true. Hence, we can solve for $\bold{a_i}|_{L_i}$ and hence $\bold{a_i}$ can be solved using the above equation.
\end{proof}

The span condition for the conventional gradient code framework in \cite{tandon2017gradient} follows from Lemma \ref{lem:two_span_b} by considering $n_2 = n_1
$ and $A_M = A$.

In order to show that the span condition in Lemma \ref{lem:two_span_b} is satisfied by the tiered gradient codes with probability 1, it is enough to show that the following support condition holds for the code under consideration and rest of the arguments follow exactly as in the proof of Lemma 3 in \cite{tandon2017gradient} (and hence omitted). We will now present the support condition for tiered gradient codes which is a sufficient condition to show that the randomized construction of the $F$ and $B_M$ matrices satisfy the span condition in Lemma \ref{lem:two_span_b}. 
\begin{condition}\label{eq:supp_cond_tgc}
	
	Consider a set of matrices $(F, \{(A_M, B_M), M \in \mathcal{M}\})$. For every $M \in \mathcal{M}$,  $I_1 \subseteq [n_1] \setminus M$, $I_2 \subseteq [n_2 - n_1]$, such that $|I_1 \cup I_2| = k-c$ and for every $\mathcal{T}_1 \subset M \cup I_1 $  and $\mathcal{T}_2 \subset I_2$ of size $|\mathcal{T}_1| + |\mathcal{T}_2| = \ell$ ($1 \leq \ell \leq k$), it needs to satisfy the following inequality for the above set of matrices to represent a tiered gradient code: $
	|\cup_{i \in \mathcal{T}_1} L_i \cup_{j \in \mathcal{T}_2} Z_j | \geq (n_1-k) + \ell.$
\end{condition}
\section{Tiered Gradient Coding}
\label{code const}

In this section, we provide our results for $(n_1,n_2,k,c)$ tiered gradient codes for the entire range of $n_1$ and $n_2$. We define the amount of computation per server as the fraction of data that is used by a server to perform computation. More the data, more is the computation time.

\begin{thm}{\label{thm: main results}}
	The amount of computation per server of $(n_1, n_2, k, c)$ tiered gradient code is as follows 
	
	\begin{enumerate}[leftmargin=.2in]
		\item For $c=1,n_1=k, n_2-n_1=1$ and  even $k$, the amount of computation per server is $\frac{2(k-1)}{k^2} $.
		\item For $c=1, k\le n_1 \leq 2(k-1) $ and $n_2 > 2(k-1)$, the amount of computation per server is $\frac{1}{2} $.
		
		

		\item For $c=1$ and $n_1,n_2\ge 3(k-1)$,  the amount of computation per server is  $\frac{n_2 -k+1 -G_1 }{n_2-G_1}$, where $G_1=	\max\{  \min\{n_2-(n_1+p^{*}),\lfloor \frac{n_1+p^{*}-k+1}{k-1} \rfloor\},$ $ \min\{n_2-n^{+},C^{*}_{n^{+}}\}, n_2-n_{min} \}, $ $p^{*}=  \lceil n_2 -n_1 - \frac{n_2-k+1}{k} \rceil$, $n^{+}=\min{\{n' \in [\max{\{n_1,n_2-6\}},n_2-1]\}}$ such that $C^{*}_{n^{+}} =\max_{n'\in [\max{\{n_1,n_2-6\}},n_2-1]} C^{*}_{n'}$, $n_{min} = \min{\{n'' \in [\max{\{n_1,n_2-6\}},n_2-1]\}} $ such that $n_2 \leq n'' + C^{*}_{n''} $. The values of $C^{*}_{n} \forall n\in [n_2] $ are provided in Table \ref{tab2}.
		
		
		\item For $c=1, 2(k-1) < n_1 < 3(k-1)$ and $n_2 \ge 3(k-1)$,  the amount of computation per server is $\frac{n_2 -k+1 -G_2 }{n_2-G_2}$, where $G_2=	\max\{ \min{\{n_2-(n_1+p^{*}),\lfloor \frac{n_1+p^{*}-k+1}{k-1} \rfloor\}},$ $ \min{\{n_2-n^{+},C^{*}_{n^{+}}\}}, n_2-n_{min} \},
		$ $ 
		p^{*}=\max{\{3(k-1)-n_1,  \lceil n_2 -n_1 - \frac{n_2-k+1}{k} \rceil\}}, n^{+}= \min{\{n' \in [\max{\{3(k-1),n_2-6\}},n_2-1]\}}$ such that $C^{*}_{n^{+}} =\max_{n'\in [\max{\{n_1,n_2-6\}},n_2-1]} C^{*}_{n'}$, $n_{min} = \min{\{n'' \in [\max{\{3(k-1),n_2-6\}},n_2-1]\}} $ such that $n_2 \le n'' + C^{*}_{n''} $. The values of $C^{*}_{n} \forall n\in [n_2] $ are provided in Table \ref{tab2}.
		
		\item For $c>1, n_1 \leq 2(k-1)$ and $n_2 > 2(k-1)$,  the amount of computation per server is $\frac{1}{2}$.

		\item For $c>1$ and $n_1\ge 2(k-1) + (k-c)$,  the amount of computation per server is $\frac{n_2-k+1-G_3}{n_2-G_3}$, where $G_3=\min{\{n_2-n_1,\lfloor \frac{n_1+p^{*}-k+c}{k-1} \rfloor \}},$ $p^{*}=\lceil n_2 -n_1 - \frac{n_2-k+c}{k} \rceil$.
		
		\item For $c>1, 2(k-1) < n_1 < 2(k-1)+(k-c)$ and $n_2 \ge 2(k-1)+(k-c)$,  the amount of computation per server is $\frac{n_2 -k+1 -G_4 }{n_2-G_4}$, where $G_4=\min{\{n_2-(2(k-1)+(k-c)),\lfloor \frac{n_1+p^{*}-k+c}{k-1} \rfloor \}},$ $p^{*}=\max{\{(2(k-1)+(k-c))-n_1,  \lceil n_2 -n_1 - \frac{n_2-k+c}{k} \rceil\}}$.

	\end{enumerate} 
\end{thm}

\begin{figure}[hbt!]
	\centering
	\includegraphics[scale=0.3]{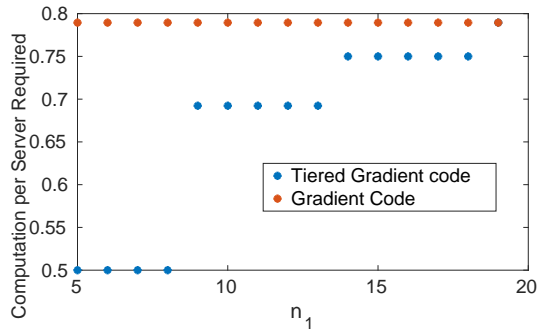}
	\caption{The amount of computation per server required as a function of $n_1$ when we vary $n_1 \in [k,n_2]$ for $n_2=19,c=1$ and $k=5$. }
	\label{fig: gain plot}
\end{figure}
We note that the amount of computation per server is $\frac{n_2-k+1}{n_2}$ for the $(n_2,k)$ gradient code, and the proposed tiered gradient codes reduce this amount to that in the statement of Theorem \ref{thm: main results} due to the flexibility of delayed start of $n_2-n_1$ servers when $c$ have finished computation. Figure \ref{fig: gain plot} illustrates the savings in the amount of computation per server for tiered gradient codes as compared to the gradient codes for $c=1$. We note that as the value of $n_1$ increase, the amount of computation per server is non-decreasing. This is because for smaller $n_1$, one possibility is to use the code construction of larger $n_1$ and only start the required number of servers initially and rest can be started when $c$ servers complete. Thus, a code for larger $n_1$ can be used for smaller $n_1$ providing more savings as $n_1$ decreases.

For all the cases discussed, the tiered gradient coding performs better than the gradient coding in terms of the amount of computations per server required. We provide the code constructions for line 2 in Section \ref{line2}, line 3 in Sections \ref{line3-1} and \ref{line3-2}, line 1 in Section \ref{line1}, and line 4 in Section \ref{sec: c1 mid range}.  Further, Section \ref{sec: c>1} provides the code construction for $c>1$ (lines 5-7).

%

%

\subsection{Tiered Gradient Codes for  $k \le n_1  \leq 2(k-1)$ }\label{line2}

In this subsection, we assume that $k \le n_1  \leq 2(k-1)$ and provide a construction of tiered gradient codes based on fractional repetition gradient codes.  The goal is to design tiered gradient codes which perform smaller computation per server than that is demanded by an $(n_2,k)$ gradient code. Here we assume that $n_2 > 2(k-1)$.


\begin{constr} \label{constr:frac1}
	We pick  $p  \triangleq 2(k-1)-n_1$. Let $Q = 2(k-1)$. 
	Consider the matrices $F_1$ and $F_2$, where $
	F_1 =   1_{\lceil \frac{n_1}{2} \rceil \times (\frac{Q}{2})}  	$ and $F_2  =   1_{\lfloor \frac{n_1}{2} \rfloor \times (\frac{Q}{2})} $. 	The matrix $F$ corresponding to the tiered gradient code is given by $ 	F = \begin{bmatrix} F_1 & 0 \\ 0 & F_2 \end{bmatrix} $. The matrix $B_M$ is as follows:
	\begin{itemize}
		\item If $M \subseteq \{1, \ldots, \lceil \frac{n_1}{2} \rceil \}$, then, $
		B_M = \begin{bmatrix}  0 & 1_{(n_2 - n_1) \times   (\frac{Q}{2})}  \end{bmatrix}$
		\item For all other possibilities of $M$, we set, $
		B_M = \begin{bmatrix}  1_{(n_2 - n_1) \times   (\frac{Q}{2})}  & 0 \end{bmatrix} $,
	\end{itemize}
	where $1_{A\times B}$ is $A\times B$ matrix with all entries as $1$. 
\end{constr}

\begin{thm}
	The code given in Construction \ref{constr:frac1} is a $(n_1,n_2,k,c)$ tiered gradient code where $k \le n_1  \leq 2(k-1)$.
\end{thm}
\begin{proof}
	To prove the theorem, we will check that the code satisfies the span condition given in Lemma \ref{lem:two_span_b}. It is clear that whenever $I_2 \neq \phi$, span condition is satisfied.
	Consider the case when $I_2 = \phi$. In this case, we need that $1_{1 \times Q}$ to lie in the span of any $k$ rows of the $F$ matrix. Since $k \geq \frac{n_1}{2} + 1 > \lceil \frac{n_1}{2} \rceil$, we will have at least one row each from the two types of rows (corresponding to the $F_1$ and $F_2$ matrices) and hence the span condition is satisfied.
\end{proof}

\textit{The proof of Theorem \ref{thm: main results} for $k \le n_1 \leq 2(k-1)$ is as follows.}

The number of samples per partition is $\frac{d}{Q}$.
The computation per server for the $(n_2,k)$ gradient code is proportional to $\frac{d }{n_2} (n_2-k+1)$. The computation per server for the tiered gradient code given in Construction  \ref{constr:frac1} is proportional to $\frac{d }{n_1+p} (n_1+p-k+1)$.
The theorem follows by noting that when $n_1 + p < n_2$,
\begin{equation*}
\frac{1}{2} = \frac{n_1+p-k+1}{n_1+p} < \frac{n_2-k+1}{n_2}.
\end{equation*}

\begin{example}
	As an example, let $n_1=7$, $k=5$, $c=1$, and $n_2=10$. In this case, the division of data is into $8$ partitions $\{D_1,D_2,...,D_8\}$. The first four servers (of the $n_1$ servers) compute the sum of gradients of the first 4 partitions ($D_1,D_2,D_3,D_4$), and the next three servers (of the $n_1$ servers) compute the sum of gradients of the last four partitions ($D_5,D_6,D_7,D_8$). 
	Suppose that server $1$ finishes the computation, three (=$n_2-n_1$) more servers are launched. These three servers compute the sum of gradients of the last four partitions ($D_5,D_6,D_7,D_8$). With the server $1$ results available which provides the sum of gradients of  the first 4 partitions, any $4$ of the remaining $9$ servers will provide the sum of gradients of the  last four partitions, thus giving the overall computation result. 
	Each server performs a computation on $4$ out of $8$ partitions, and thus on $1/2$ of the data. This is in contrast to each server performing computation on $6/10$ of the data in case of the $(n_2,k)$ gradient code. As $n_2$ increases, the relative improvement of the tiered gradient codes increases. 
\end{example}


\subsection{Tiered  Gradient Codes for $n_1  \geq 3(k-1)$}\label{line3-1}

In this section, we construct tiered gradient codes in which a subset of servers under consideration will be allocated a cyclic repetition gradient code of suitable parameters. We will construct codes for all values of $n_2$,  with $n_1-k+1$ computations per server. We note that this is the best possible, since the gradient code restricted to first $n_1$ servers has to be  an $(n_1,k)$ gradient code and $n_1-k +1$ is the lower bound on the number of computations per server of an $(n_1,k)$ gradient code \cite{tandon2017gradient}. We make the following observation with respect to the condition which the tiered gradient code has to necessarily satisfy. These will in turn be used to construct certain tiered gradient codes. 

\begin{lem} \label{lem:supp_cond}
	Consider an $(n_1, n_2, k,c=1)$ tiered gradient code. Suppose $Q=n_1$ and the support of the $F$ matrix is picked as those given by the cyclic repetition gradient code. Let the $m^{\text{th}}$ server finish its job in the first phase, i.e. $M = \{m\}$. Let $L_m$ denote the support of $\bold{f}_m$, and $Z_j$ denote the support of $\bold{b}_j$ ($\bold{b}_j$ is the $j^{\text{th}}$ row of matrix $B_M$). Then, the following holds: 
	$
	[ 0 , n_1-1] \setminus L_m \subset Z_j, \forall j \in [n_2-n_1].
	$
\end{lem}

\begin{proof}
	Suppose not. Consider $r \in \{ 0 , 1, \ldots, n_1-1\} \setminus L_m$ and $r \notin Z_j$ for some $j \in [n_2-n_1]$. Based on the structure of the cyclic repetition gradient code, there are $k-1$ rows in the $F$ matrix including $m$ where the $r^{\text{th}}$ column is zero. Considering these $k-1$ rows and adding the $j^{\text{th}}$ row of the $B_M$ matrix, we have a set of $k$ rows which is required to satisfy the span property. However, since the $r^{\text{th}}$ coordinate is zero in all these rows, $1_{1 \times Q}$ cannot be in the span of these $k$ rows. Hence, Lemma \ref{lem:supp_cond} is necessary for the code to be $(n_1,n_2,k,c=1)$ gradient code.
\end{proof}

We consider $Q=n_1$ and $C_{n_1}= \lfloor \frac{n_1-k+1}{k-1} \rfloor$. Initially, the first $n_1$ servers are launched. We assume (without loss of generality) that the server $1$ has finished the job in the first phase. Then the remaining $n_2-n_1$ servers are launched. 
We will now construct codes for the case where $n_2 = n_1 + C_{n_1}$. Let $B_M$ be a $C_{n_1} \times n_1$ matrix with $\bold{b}_i$ representing the $i^{th}$ row and $Z_{i}$ representing the support of $\bold{b}_i$, where $i \in [C_{n_1}]$. Let the columns of the $F$ and $B_M$ matrices be indexed by $[0,n_1-1]$.
\begin{constr}[$n_2 = n_1 + C_{n_1}$] \label{constr:cyclic_gen_1}
	The support structure of the matrix $F$ is as follows:
	\begin{equation*}
	supp(\bold{f_i}) = [i-1,i+(n_1-k-1)] \mod n_1.
	\end{equation*}
	
	The procedure to design the support of each row of the $B_M$ matrix is as follows.
	
	If $C_{n_1} = 1$, the $B_M$ matrix is a row matrix. The $k-1$ coordinates of $Z_1$ are given by $	[n_1-k+1,n_1-1]  \subset Z_1.$ 	We pick the remaining $n_1 - 2(k-1)$ coordinates as a subset of $L_1 = [0, n-k]$ such that at least one of every pair of consecutive coordinates (modulo $n_1$) is present in the set. 
	
	If $C_{n_1} > 1$, do the following. Let  $
	l=(n_1-k+1) - (k-1)C_{n_1}.$ The $l+k-1$ coordinates of $Z_j, j \in [C_{n_1}]$, are given by $
	[n_1-(l+k-1),n_1-1]  \subset Z_j, j=[C_{n_1}]. $ 	Let  
	\[
	B_{M}=
	\begin{bmatrix}
	B_{M_1} & B_{M_2} & \ldots &  B_{M_{k-1} } & B_{M'}
	\end{bmatrix}
	\]
	Each submatrix $B_{M_j}$, where $j \in [k-1]$, is of size $C_{n_1} \times C_{n_1}$ and the $B_{M'}$ matrix is of size $C_{n_1} \times (l+k-1)$. The $B_{M'}$ matrix constitutes the $l+k-1$ columns - $ [n_1-(l+k-1),n_1-1]$ of the $B_M$ matrix filled with non zero entries. The support structure for the remaining coordinates of the $B_M$ matrix is obtained from the design of the support structure corresponding to the matrices $B_{M_j}, j \in [k-1]$. The support of the $i^{th}$ row of each matrix $B_{M_j}$, where $j \in [k-1]$, is of the form $[i-1, i+C_{n_1}-3] $ mod $C_{n_1}$.

	Now, we will present a randomized construction of the matrices $F$ and $B_M$. The matrix $H$ of size $(n_1-k) \times n_1$ is picked at random as given in Section \ref{sec:rev_cyclic}.
	Each vector $\bold{f_i}$ is calculated by solving the following equation
	\begin{equation*}
	\bold{f_i}|_{L_i } H(:,L_i)^T = 0,
	\end{equation*}
	$\bold{b}_i$ is calculated by solving the following equation
	\begin{equation*}
	\bold{b}_i|_{Z_i } H(:,Z_i)^T = 0.
	\end{equation*}
\end{constr}

\begin{proof}
	Initially, the first $n_1$ servers are launched. Without loss of generality, let us assume that the server $1$ finishes the job first. Then the remaining $n_2 - n_1$ servers are launched. The procedure to design the support of each row of the $B_M$ matrix is as follows. It is necessary that $|Z_{j}| = n_1-k+1$, for each $j \in [C_{n_1}]$. From Lemma \ref{lem:supp_cond}, we have
	$
	[0, n_1-1] \setminus L_1  \subset Z_j, \ \ j=[C_{n_1}]
	$ and
	$
	|[ 0, n_1-1] \setminus L_1| = k-1.
	$
	Thus $k-1$ coordinates are included in each $Z_j,j=[C_{n_1}]$. We have to add exactly $|Z_j| - (k-1) = n_1 - 2(k-1) > 0$ coordinates from $L_1$ to the set to complete the specification of $Z_j$.  We pick these $n_1 - 2(k-1)$  coordinates as a subset of $L_1 = [0, n_1-k]$ such that at least one of every pair of consecutive coordinates (modulo $n_1$) is present in the set. We will refer to this condition as consecutive coordinate property. It is possible to pick $n_1 - 2(k-1)$ coordinates satisfying the consecutive coordinate property only if $n_1 - 2(k-1) \geq \lfloor \frac{n_1 -k+1}{2} \rfloor$. We can easily see that the above property is satisfied when $n_1 \geq 3(k-1)$. In addition, these are also picked so that $|Z_j \cup Z_i| = n_1$, for any $j,i \in [C_{n_1}]$.

	The $l$ coordinates from  $L_1$ - $[(n_1-k) -(l-1),(n_1-k)]$ are also included in $Z_j$. Thus, totally, $l+k-1$ coordinates are included in each $Z_j$. We have to add $|Z_j|-(l+k-1) = n_1-2(k-1)-l$ more coordinates to $Z_j$ from $L_1 \setminus [(n_1-k) -(l-1), (n_1-k) ] = [0, n_1-k -l]$. That is, we need to pick $n_1-2(k-1)-l$ from $n_1-k+1-l$ locations available.

	Let $B_{M'}$ be the matrix obtained by taking the $l+k-1$ coordinates corresponding to $ [n_1-(l+k-1),n_1-1 ]$ from each row in the $B_M$ matrix, i.e., by taking the last $l+k-1$ columns -$ [n_1-(l+k-1),(n_1-1) ]$ from the $B_M$ matrix. 
	$B_{M_1}$ constitutes the first $C_{n_1}$ columns of the $B_{M}$ matrix, $B_{M_2}$ constitutes the next $C_{n_1}$ columns and so on.
	Hence, each $B_{M_j}$, $j \in [k-1]$, is a $C_{n_1} \times C_{n_1}$ matrix which is obtained by taking distinct and consecutive $C_{n_1}$ columns from the $B_M$ matrix sequentially. $n_1-2(k-1)-l$ more coordinates to be added to $Z_j,j=[C_{n_1}]$ is obtained from the design of the support structure corresponding to the matrices $B_{M_j}, j \in [k-1]$.
	We have the support structure of the $B_{M'}$ matrix. The support structure for the remaining coordinates of the $B_M$ matrix is obtained from the design of the support structure corresponding to the matrices $B_{M_j}, j \in [k-1]$.
	
	The support of the $i^{th}$ row of each matrix $B_{M_j}$, where $j \in [k-1]$, is of the form $[i-1, i+C_{n_1}-3]$ mod $C_{n_1}$. The cardinality of the support of each row of the $B_{M_j}$ matrix is $C_{n_1}-1$, i.e, there is exactly one zero in each row of the $B_{M_j}$ matrix at disjoint locations. Hence, the number of zeros in each row of the $B_{M}$ matrix is exactly $k-1$, which is exactly what we needed. 
	The cardinality of the support of union of any two rows of  the $B_{M_j}$ matrix is $C_{n_1}$. Hence if we take union of support of any two rows in the $B_{M}$ matrix, then it has cardinality $n_1$. That is, $|Z_{r} \cup Z_{s}| = n_1$, for any $r,s \in [C_{n_1}]$. Hence the support structure of the $B_M $ matrix satisfies all the required conditions.
\end{proof}

\begin{thm} \label{thm:cyclic1}
	The code given in Construction \ref{constr:cyclic_gen_1} is a $(n_1,n_2=n_1+C_{n_1},k,c=1)$ tiered gradient code where $n_1 \geq 3(k-1)$ and $C_{n_1}=\frac{n_1-k+1}{k-1}$.
\end{thm}
\begin{proof}
	We have to show that Condition \ref{eq:supp_cond_tgc} is satisfied by the code in Construction \ref{constr:cyclic_gen_1} with probability 1. Here $M=\{1\}$, assuming that the server $1$ finished its task first. 	If $\mathcal{T}_2 = \phi$, the above condition follows from the support structure of the cyclic repetition code. If $|\mathcal{T}_2| = 1$ and $M \in \mathcal{T}_1$, then the support of the union of $\mathcal{T}_1$ and $\mathcal{T}_2$ is $[n_1]$ and hence Condition \ref{eq:supp_cond_tgc} is satisfied.
	Now, we will consider the case when $|\mathcal{T}_2| = 1$, $M \notin \mathcal{T}_1$, $|\mathcal{T}_1| = k-2$ and $\mathcal{T}_1$ is such that $
	|\cup_{i \in \mathcal{T}_1} L_i | = n_1 - k + (k-2) = n_1-2.
	$
	Based on the cyclic support structure of the $F$ matrix, the above condition is true whenever $(k-2)$ consecutive rows (modulo $n_1$) are picked. Hence, the two coordinates which are not included in the union are consecutive. Since the support of rows in $B_M$ matrix satisfies consecutive coordinate property, at least one of the coordinates of the two coordinates which are not picked up before will be included after adding the new row. So we have, $
	|\cup_{i \in \mathcal{T}_1} L_i \cup_{j \in \mathcal{T}_2} Z_j | \geq (n_1-1).$ 	
	Hence Condition \ref{eq:supp_cond_tgc} is satisfied.
	
	For the cases when $|\mathcal{T}_2| \geq 2$, since $Z_i$ and $Z_j$ are chosen such that $Z_i \cup Z_j = [n_1]$, for any $i,j \in [C_{n_1}]$, we have that the condition \ref{eq:supp_cond_tgc} being trivially satisfied. 
	Thus, Condition \ref{eq:supp_cond_tgc} is satisfied for all cases and hence the code is a $(n_1,n_2=n_1+C_{n_1},k,c=1)$ tiered gradient code.
\end{proof}

\begin{example} {\label{exmp: 1}}
	Let $n_1=9,k=3,c=1$ and $n_2=12.$ We split data into $9$ partitions -$\{x_0,x_1,\ldots , x_8\}$. The server $i$ is assigned data $\{x_j, j \in [i-1, i+5]\}$. Each server computes the gradients on their respective data. We assume that server $1$ finishes its task first when $n_1$ servers are launched. After that the remaining $n_2-n_1=3$ servers are launched. Since server $1$ doesn't have $\{x_7,x_8\}$ as its content, we have to include  $\{x_7,x_8\}$ in the content of three added servers. Here, $l=1$. Hence, $x_6$ needs to be included in the content of three added servers. The last three columns of the $B_M$ matrix are filled with non zero entries.  Let $
	B_{M}=
	\begin{bmatrix}
	B_{M_1} & B_{M_2} & B_{M'}
	\end{bmatrix}
	$. 	The $B_{M'}$ matrix is obtained by taking the last three columns of $B_M$. Hence it is a $3 \times 3 $ matrix.
	Both the $B_{M_1}$ and $B_{M_2}$ matrices are $3 \times 3$ matrices. $B_{M_1}$ is the submatrix formed by the first three columns of the $B_{M}$ matrix and $B_{M_2}$ is formed by the next three columns of the $B_{M}$ matrix.  The support of the $i^{th}$ row of each of the matrices $B_{M_1}$ and $B_{M_2}$ is of the form $\{i-1,i\}$ mod $3$. Hence the structures of $B_{M_1},B_{M_2} $ and $B_M$ are of the form $
	B_{M_1}= B_{M_2}=
	\begin{bmatrix}
	* & 0 & *  \\
	* & * & 0  \\
	0 & * & *  \\
	\end{bmatrix}
	$ and $
	B_{M}=
	\begin{bmatrix}
	* & 0 & * & * & 0 & * & * & * & *\\
	* & * & 0 & * & * & 0  & * & * & *\\
	0 & * & * & 0 & * & *  & * & * & *\\
	\end{bmatrix}
	$. 
	The symbol $*$ in the above matrices implies non zero entries in those locations. 
	The content of the three added servers are $\{x_0,x_2,x_3,x_5,x_6,x_7,x_8\},\{x_0,x_1,x_3,x_4,x_6,x_7,x_8\}$ and $\{x_1,x_2,x_4,x_5,x_6,x_7,x_8\}$ respectively. Out of the $8$ servers which haven't finished the job earlier and the three added servers, any two servers can give the sum of the gradients along with server $1$. Each server does $7/9$ computations compared to $10/12$ required for the $(n_2,k)$ gradient code.
\end{example}


%
\begin{note}
	If $n_1 <  n_2 < n_1 + C_{n_1}$, we take the support structure of any $n_2-n_1$ rows of the $B_M$ matrix constructed using Construction \ref{constr:cyclic_gen_1} ($ n_2 = n_1 +C_{n_1}$) to generate the support structure for the $B_M$ matrix in this case. The support structure of the matrix $F$, the construction of the $B_M$ and $F$ matrices using the above support structures are same as in Construction \ref{constr:cyclic_gen_1}.
\end{note}

\subsection{Tiered Gradient Codes for $n_1=k,n_2-n_1=1$ (even $k$)} \label{line1}

\label{sec: n-=k}
In this section, we provide tiered gradient codes for $n_1 = k, n_2-n_1=1$, where $k$ is even. The computation per server required is $\frac{2(k-1)}{k^2}$. 
Let $t=k-1$. We split the data into $\frac{k^2}{2}$ partitions. Each user is assigned $t$ partitions of data. The code construction is as follows.
\begin{constr} {\label{constr: n1 = k}}
	($n_1 = k, n_2-n_1=1,k$ even). The support structure of the matrix $F$ is as follows: $
	supp(\bold{f_i}) = [(i-1)\left \lceil \frac{t}{2} \right \rceil, (i-1)\left \lceil \frac{t}{2} \right \rceil + (t-1) ]  \mod n_1.
	$ 
	If server $ m \in [1,n_1]$ finishes the task first in the first phase, the support of the $B_{M}$ matrix, which is a row vector is as follows: $
	Z_1 = \bigcup_{j \in [0,k-1] \backslash (m-1)} (j-1)\left \lceil \frac{t}{2} \right \rceil + t \mod n_1,
	$
	
\end{constr}
\begin{proof}
	We split data into $\frac{k^2}{2}$ parts, namely $\{x_0,x_1, \ldots ,x_{\frac{k^2}{2}-1}\}$. 
	The support of the first row of the $F$ matrix is $[0 ,t-1]$. Each row of the $F$ matrix is obtained by shifting the previous row by $ \left \lceil \frac{t}{2} \right \rceil$ towards right. Any two consecutive servers have exactly $\left \lfloor \frac{t}{2} \right \rfloor$ partitions of data in common. Server $i$ and $i+1$ have $\{x_{(i)\left \lceil \frac{t}{2} \right \rceil }, x_{(i)\left \lceil \frac{t}{2} \right \rceil +1}, \ldots ,  x_{(i+1)\left \lceil \frac{t}{2} \right \rceil -1}  \}$ in common. 
	Initially the first $n_1$ servers are launched. Let us assume that server $m$ finishes the task first. Then one more server is launched. The content of this server includes the partitions of data which are unique to each of the first $n_1$ servers except server $m$. The data which is unique to the server $i, i \in [0,n_1]$ is $x_{(i-1)\left \lceil \frac{t}{2} \right \rceil + \left \lfloor \frac{t}{2} \right \rfloor}$. Hence, $k-1 = t$ partitions of data are included in the newly added server.
	
\end{proof}

\begin{thm}
	The code given in Construction \ref{constr: n1 = k} is a ($n_1,n_2=n_1+1,k=n_1,c=1$) tiered gradient code where $k$ is even.
\end{thm}

\begin{proof}
	
	We need to prove that the support condition given in Condition \ref{eq:supp_cond_tgc} is satisfied by the code in Construction \ref{constr: n1 = k}. Here, $M=\{m\}$. The $F$ matrix is a circulant matrix with each row shifted by $ \left \lceil \frac{t}{2} \right \rceil$ towards right from the previous row. Hence, if $\mathcal{T}_2 = \phi$, Condition \ref{eq:supp_cond_tgc} holds. 	We will now consider the case when $|\mathcal{T}_2| = 1$, $m \notin \mathcal{T}_1$, $|\mathcal{T}_1| = k-2$. That is precisely when we have taken all the servers from the first $n_1$ servers except server $m$ and one more server which is referred as server $b$. 
	When we picked the coordinates for the server $n_2$, we have included the coordinate which is unique to the server $b$. Hence, Condition \ref{eq:supp_cond_tgc} is satisfied. Thus, Condition \ref{eq:supp_cond_tgc} is satisfied for all cases and hence the code is a $(n_1,n_2=n_1+1,k=n_1,c=1)$ tiered gradient code for even $k$.
\end{proof}

\begin{table*}
	\begin{center}
		\begin{tabular}{ |p{0.5cm}| p{7cm} | p{7cm} | }
			\hline
			
			\textit{$C^{*}_{n_1}$}	&\textit{All the possible cases when $p$ is even}& \textit{All the possible cases when $p$ is odd}   \\ \hline
			
			$2$&	$p=0$ & $p'> \frac{3(p+1)}{2} $  \\ \cline{2-2}
			&  $p'> 3P$	&  \\ \hline

			$3$& $\frac{3p}{2}\leq p'\leq 3p$ &$\frac{3(p-1)}{2} \leq p'\leq \frac{3(p+1)}{2}$ \\ \hline

			&   $p'=1$ mod $3$ and  &  $p'=1$ mod $3$ and $  3 \lfloor \frac{p-1}{4} \rfloor < p' \leq \frac{3(p-1)}{2}$  \\ \cline{3-3}
			$4$&  $ \max{\{0,3 \lfloor \frac{p}{4} \rfloor -1\}} < p' < \frac{3p}{2}$ 	 &  $p'=2$  mod $3$, $ 2 \leq p'\leq  \lceil \frac{3(p-1)}{2} \rceil$ and \\ && $p' \neq 3 \lfloor \frac{p-1}{4} \rfloor -1$  \\ \cline{2-3}	 		
			& $p'=0,2$  mod $3$ and $ 2< p'\leq  \lceil \frac{p-1}{2} \rceil$ 	  & $p'=2$  mod $3$, $p'=7$ mod $8$ and \\ && $ \frac{3(p-7)}{4} < p' < \frac{3(p-1)}{2}$ \\ \cline{2-3}
			&   $p'=0,2$  mod $3$ and $ 3 \lceil \frac{p}{4} \rceil -1 < p' < \frac{3p}{2}$  & $p'=2$  mod $3$, $p' \neq 7$ mod $8$ and \\ && $6 \lfloor \frac{p}{8} \rfloor -3 < p' < \frac{3(p-1)}{2}$ \\ \hline
			
			&  & $p'=1$ mod $3$ and $ 3 \lceil \frac{p-2}{6} \rceil  < p' \leq 3 \lfloor \frac{p-1}{4} \rfloor$\\ \cline{3-3}
			$5$ && $p'=2$  mod $3$, $ 2 \leq p'\leq  \lceil \frac{3(p-1)}{2} \rceil$ and \\ & $p'$ is a multiple of $3$ and $\lceil \frac{p-1}{2} \rceil < p' <\frac{3p}{2}$ &$p' = 3 \lfloor \frac{p-1}{4} \rfloor -1$ \\ \cline{3-3}
			&  & $p'=0$  mod $3$, $p'=7$ mod $8$ and \\ &&$0< p' \leq \frac{3(p-7)}{4}$ \\ \cline{3-3}
			&  & $0< p' \leq 6 \lfloor \frac{p}{8} \rfloor -3 $ \\ \hline
			
			$6$& $p'=2$ mod $3$ and $\lceil \frac{p-1}{2} \rceil < p' <\frac{3p}{2}$& $p'=1$ mod $3$ and $0< p' \leq 3 \lceil \frac{p-2}{6} \rceil$\\ \cline{2-2}
			&$p'=0$ mod $3$ and $0< p'< \max{\{0,3 \lfloor \frac{p}{4} \rfloor -1\}}$ & \\ \hline
			$0$ & \multicolumn{2}{c|}{For all other cases not discussed above}\\ \hline

		\end{tabular}
	\end{center}
	\caption{Table that illustrates the value of $C^{*}_{n_1}$ for any $k \geq p+4$ and $n_1 = 3(k-1) + p$, for some integer $p$. Let $k'=p+4$ and $n'_1=3(k'-1)+p$.  If $k \geq p+4$ and $n_1 = 3(k-1) + p$, we can write $k$ and $n_1$ in terms of $k'$ and $n'_1$ as $k=k'+p'$ and $n_1 = n'_1 +3p',$ where $p'=\{0,1,2,...\}$.}
	\label{tab2}
\end{table*}

\subsection{Tiered Gradient Code for $k \geq 4,n_1 \in [3(k-1), 3(k-1)+(k-4)]$.}\label{line3-2}

In this subsection, we consider the case where $k \geq 4$ and $n_1 \in [3(k-1), 3(k-1)+(k-4)]$. For such cases we provide construction for $n_2=n_1+C^{*}_{n_1}$, where $C^{*}_{n_1} \geq C_{n_1}$. The value of $C^{*}_{n_1}$ is given in Table \ref{tab2}.

For any integer $p$, if $k \geq p+4$ and $n_1 = 3(k-1) + p$, the code construction is provided below. Let $k'=p+4$ and $n'_1=3(k'-1)+p$, for some integer $p$. We can write $k$ and $n_1$ in terms of $k'$ and $n'_1$ as $k=k'+p'$ and $n_1 = n'_1 +3p',$ where $p'=\{0,1,2,...\}$. The value $C^{*}_{n_1}$ varies from $2$ to $6$ depending upon $k$ and $n_1$, which is given in Table \ref{tab2}. Let $(0**)^y$ represent the sequence $\{0**\} $ repeated $y$ times, i.e,
\begin{equation}
(0**)^y=  \begin{bmatrix} 
0 &*& *& 0& *& *&  .... &0& *& * 
\end{bmatrix}_{1 \times 3y}
\end{equation}
where $*$ represents some non zero entry.  Similarly let $(**0)^y$ represent the sequence $ \{**0\}$, $(0*)^y$ represent the sequence $\{0 *\} $, $(*0)^y$ represent the sequence $\{*0\}$ and $(*)^y$ represent the sequence $\{*\}$ repeated $y$ times. We will now construct codes for the case where $n_2 = n_1 + C^{*}_{n_1}$. Let $B_M$ is a $C^{*}_{n_1} \times n_1$ matrix with $\bold{b}_i$ representing the $i^{th}$ row and $Z_{i}$ representing the support of $\bold{b}_i$, where $i \in [C^{*}_{n_1}]$. 
\begin{constr}[$n_2 = n_1 + C^{*}_{n_1}$] \label{constr:cyclic_spec_k_n1_equal Cn1}
	The support structure of the matrix $F$ is as follows:
	\begin{equation*}
	supp(\bold{f_i}) = [i-1, i+(n_1-k-1)] \mod n_1.
	\end{equation*}
	If $C^{*}_{n_1} = 2$, the support structure of  the $B_M$ matrix is same as in Construction \ref{constr:cyclic_gen_1}.
	The support structure of  the $B_M$ matrix for all other values of  $C^{*}_{n_1} $ is given in Table \ref{tab3}. The construction of the $B_M$ and $F$ matrices using the above support structures are same as in Construction \ref{constr:cyclic_gen_1}.
\end{constr}

\begin{table*}
	\begin{center}
		\begin{tabular}{|p{0.4cm}| p{0.6cm} | p{7.5cm} | p{8cm}| }
			\hline
			
			\textit{$C^{*}_{n_1} $}   &   \textit{$p'$ } &   \textit{The support structure of the $B_M$ matrix if $p$ is even}  &  \textit{The support structure of the $B_M$ matrix if $p$ is odd} \\
			\hline
			$3$   &  any $p'$   &  
			$ \begin{bmatrix} 
			(0**)^{p+1} (0*)^{n_1} 0 (*)^{k-1}\\
			0 (*0)^{n_1} (**0)^{p+1}(*)^{k-1} \\
			* (0*)^{\lceil \frac{n_1}{2} \rceil +1} (0**)^{p-1} (0*)^{\lfloor \frac{n_1}{2} \rfloor +2} (*)^{k-1}
			\end{bmatrix}  $ &	$ \begin{bmatrix} 
			(0**)^{p+1} (0*)^{n_1} 0 (*)^{k-1}\\
			0 (*0)^{n_1} (**0)^{p+1} (*)^{k-1}\\
			* (0*)^{\lceil \frac{n_1}{2} \rceil +1} (0**)^{p-1} (0*)^{\lfloor \frac{n_1}{2} \rfloor +2} (*)^{k-1}
			\end{bmatrix}  $   \\
			\hline
			$4$   &   $1,2$ mod $3$    &  	$ \begin{bmatrix} 
			(0**)^{p+1} (0*)^{n_1} 0 (*)^{k-1} \\
			0 (*0)^{n_1} (**0)^{p+1} (*)^{k-1} \\
			(0**)^{p} (0*)^{n_1+2}  (*)^{k-1}\\
			(*0)^{n_1+2} (**0)^{p} (*)^{k-1}\\
			\end{bmatrix}  $    &  	$ \begin{bmatrix} 
			(0**)^{p+1} (0*)^{n_1} 0 (*)^{k-1} \\
			0 (*0)^{n_1} (**0)^{p+1} (*)^{k-1} \\
			(0**)^{p} (0*)^{n_1+2}  (*)^{k-1}\\
			(*0)^{n_1+2} (**0)^{p} (*)^{k-1}\\
			\end{bmatrix} $           \\ \hline
			$4$ & $0$ mod $3$  & $ \begin{bmatrix} 
			(0**)^{p+1} (0*)^{n_1} 0 (*)^{k-1}\\
			0 (*0)^{n_1} (**0)^{p+1} (*)^{k-1}\\
			0 * (**0)^{\frac{p}{2}-1} (*0)^{n_1+1} * (0**)^{\frac{p}{2}} * 0 (*)^{k-1}  \\
			0 * (**0)^{\frac{p}{2}} * (0*)^{n_1+1}  (0**)^{\frac{p}{2}-1} * 0 (*)^{k-1}  
			\end{bmatrix}$ if $  p' \in [3, \lceil \frac{p-1}{2} \rceil]$& $\begin{bmatrix} 
			(0**)^{p+1} (0*)^{n_1} 0 (*)^{k-1}\\
			0 (*0)^{n_1} (**0)^{p+1} (*)^{k-1}\\
			0 * (**0)^{p-1} (*0)^{n_1+2} * (*)^{k-1} \\
			* (0*)^{n_1+2} (0**)^{p-1} * 0 (*)^{k-1}
			\end{bmatrix}$ \\ \hline
			$5$ & $1,2$ mod $3$ & $ \begin{bmatrix} 
			(0**)^{p+1} (0*)^{n_1} 0 (*)^{k-1}\\
			0 (*0)^{n_1} (**0)^{p+1} (*)^{k-1} \\
			0 * (**0)^{\frac{p}{2}-1} (*0)^{n_1+1} * (0**)^{\frac{p}{2}} * 0 (*)^{k-1}  \\
			0 * (**0)^{\frac{p}{2}} * (0*)^{n_1+1}  (0**)^{\frac{p}{2}-1} * 0 (*)^{k-1}  \\
			* (0*)^{\lceil \frac{n_1}{2} \rceil +1} (0**)^{p-1} (0*)^{\lfloor \frac{n_1}{2} \rfloor +2} (*)^{k-1}
			\end{bmatrix}$  & $ \begin{bmatrix} 
			(0**)^{p+1} (0*)^{n_1} 0 (*)^{k-1} \\
			0 (*0)^{n_1} (**0)^{p+1} (*)^{k-1}\\
			(0**)^{p} (0*)^{n_1+2}  (*)^{k-1}\\
			(*0)^{n_1+2} (**0)^{p} (*)^{k-1}\\
			0 * (**0)^{\frac{p-1}{2}} (*0)^{n_1+1} * (0**)^{\frac{p-1}{2}} * 0 (*)^{k-1}
			
		\end{bmatrix}$ \\ \hline
		
		$5$ & $0$ mod $3$ & $ \begin{bmatrix} 
		(0**)^{p+1} (0*)^{n_1} 0 (*)^{k-1}\\
		0 (*0)^{n_1} (**0)^{p+1} (*)^{k-1} \\
		0 * (**0)^{\frac{p}{2}-1} (*0)^{n_1+1} * (0**)^{\frac{p}{2}} * 0 (*)^{k-1}  \\
		0 * (**0)^{\frac{p}{2}} * (0*)^{n_1+1}  (0**)^{\frac{p}{2}-1} * 0 (*)^{k-1}  \\
		* (0*)^{\lceil \frac{n_1}{2} \rceil +1} (0**)^{p-1} (0*)^{\lfloor \frac{n_1}{2} \rfloor +2} (*)^{k-1}
		\end{bmatrix}$  &$ \begin{bmatrix} 
		(0**)^{p+1} (0*)^{n_1} 0 (*)^{k-1}\\
		0 (*0)^{n_1} (**0)^{p+1} (*)^{k-1}\\
		* (0*)^{\lceil \frac{n_1}{2} \rceil +1} (0**)^{p-1} (0*)^{\lfloor \frac{n_1}{2} \rfloor +2} (*)^{k-1}\\
		0 * (**0)^{u} (*0)^{n_1+1} 1 (0**)^{p-1-u} * 0 (*)^{k-1}  \\
		0 * (**0)^{p-1-u} (*0)^{n_1+1} 1 (0**)^{u} * 0 (*)^{k-1}
		\end{bmatrix},$ 	
		where 	$
		u=
		\begin{cases}
		\frac{p'-7}{3} & \text{if}\ p' = 1 \text{ mod } 6 \\
		\frac{p'-3}{3} & \text{if}\ p' = 3 \text{ mod } 6 \\
		\frac{p'-5}{3} & \text{if}\ p' = 5 \text{ mod } 6 
		\end{cases}$  \\ \hline
		
		$6$ & any $p'$ & $ \begin{bmatrix} 
		(0**)^{p+1} (0*)^{n_1} 0 (*)^{k-1} \\
		0 (*0)^{n_1} (**0)^{p+1} (*)^{k-1} \\
		(0**)^{p} (0*)^{n_1+2}  (*)^{k-1} \\
		(*0)^{n_1+2} (**0)^{p} (*)^{k-1}\\
		0 * (**0)^{\frac{p}{2}-1} (*0)^{n_1+1} * (0**)^{\frac{p}{2}} * 0 (*)^{k-1}  \\
		0 * (**0)^{\frac{p}{2}} * (0*)^{n_1+1}  (0**)^{\frac{p}{2}-1} * 0 (*)^{k-1}  
		\end{bmatrix}$  & $ \begin{bmatrix} 
		(0**)^{p+1} (0*)^{n_1} 0 (*)^{k-1}\\
		0 (*0)^{n_1} (**0)^{p+1} (*)^{k-1}\\
		(0**)^{p} (0*)^{n_1+2}  (*)^{k-1}\\
		(*0)^{n_1+2} (**0)^{p} (*)^{k-1}\\
		0 * (**0)^{u} (*0)^{n_1+1} * (0**)^{p-1-u} * 0 (*)^{k-1}  \\
		0 * (**0)^{p-1-u} (*0)^{n_1+1} * (0**)^{u} * 0 (*)^{k-1}
		\end{bmatrix}$, where 	$u=
		\begin{cases}
		\frac{p'+2}{3} & \text{if}\ p' = 1 \text{ mod } 6 \\
		\frac{p'-3}{3} & \text{if}\ p' = 3 \text{ mod } 6 \\
		\frac{p'-2}{3} & \text{if}\ p' = 5 \text{ mod } 6 
		\end{cases}$
		\\ \hline
	\end{tabular}
	
\end{center}
\caption{Table that illustrates the support structure of the $B_M $ matrix associated with Construction \ref{constr:cyclic_spec_k_n1_equal Cn1}.}
\label{tab3}

\end{table*}

\textit{Discussion on Construction \ref{constr:cyclic_spec_k_n1_equal Cn1}: }

The procedure to design the support of each row of the $B_M$ matrix is as follows.
$k-1$ coordinates of $Z_j$ are given by $ 	[n_1-k+1,n_1-1]  \subset Z_j,j \in C^{*}_{n_1},$ which corresponds to those partitions which are not included in server $1$.

We pick the remaining $n_1 - 2(k-1)$ coordinates as a subset of $L_j = [0, n-k]$ such that at least one of every pair of consecutive coordinates is present in the set. This is termed as the consecutive coordinate property. And also, the support structure of the $B_M$ matrix is designed in such a way that for every $\mathcal{T}_2$, which is a subset of $C^{*}_{n_1}$ of size $\ell$ ($1 \leq \ell \leq k-1$), we have $
|\cup_{i \in \mathcal{T}_2} Z_i| \geq (n_2-k) + \ell.
$

\begin{thm} \label{thm:special k and n1}
The code given in Construction \ref{constr:cyclic_spec_k_n1_equal Cn1} is a $(n_1,n_2=n_1+C^{*}_{n_1},k,c=1)$ tiered gradient code where $k \geq 4,n_1 \in [3(k-1), 3(k-1)+(k-4)]$.
\end{thm}
\begin{proof}
We have to show that Condition \ref{eq:supp_cond_tgc} is satisfied by the code in Construction \ref{constr:cyclic_spec_k_n1_equal Cn1}. Here $M=\{1\}$, assuming that server $1$ finished its task first. 	If $\mathcal{T}_2 = \phi$, Condition \ref{eq:supp_cond_tgc} follows from the support structure of cyclic repetition code. If $|\mathcal{T}_2| = 1$ and $M \in \mathcal{T}_1$, then the support of the union of $\mathcal{T}_1$ and $\mathcal{T}_2$ is $[n_1]$ and hence Condition \ref{eq:supp_cond_tgc} is satisfied.
Now, we will consider the case when $|\mathcal{T}_2| = 1$, $M \notin \mathcal{T}_1$, $|\mathcal{T}_1| = k-2$ and $\mathcal{T}_1$ is such that $
|\cup_{i \in \mathcal{T}_1} L_i | = n_1 - k + (k-2) = n_1-2.
$
Based on the cyclic support structure of the $F$ matrix, Condition \ref{eq:supp_cond_tgc} is true whenever $(k-2)$ consecutive rows (modulo $n_1$) are picked. Hence, the two coordinates which are not included in the union are consecutive. Since the support of the rows in the $B_M$ matrix satisfies consecutive coordinate property, at least one of the coordinates of the two coordinates which are not picked up before will be included after adding the new row. So we have, $
|\cup_{i \in \mathcal{T}_1} L_i \cup_{j \in \mathcal{T}_2} Z_j | \geq (n_1-1). $
Hence Condition \ref{eq:supp_cond_tgc} is satisfied.  The support structure of the $B_M$ matrix is designed in such a way that Condition \ref{eq:supp_cond_tgc} is satisfied. Hence, for the cases when $|\mathcal{T}_2| \geq 2$, Condition \ref{eq:supp_cond_tgc} is satisfied.

Thus, Condition \ref{eq:supp_cond_tgc} is satisfied for all cases and hence the code is a $(n_1,n_2=n_1+C^{*}_{n_1},k,c=1)$ tiered gradient code.
\end{proof}

\begin{constr}\label{const: gen_c1_n1>3(k-1)}
\textit(General $n_1,n_2 \ge 3(k-1)$).
If $n_1,n_2 \ge 3(k-1),$

\begin{enumerate}
	\item we need to find a $p$ such that the following condition is satisfied, i.e., for $C_{n_1} =\lfloor \frac{n_1-k+1}{k-1} \rfloor$,
	\begin{align*}
	n_2 &\leq n_1 +p + C_{n_1+p}\\
	& =n_1+p+\frac{n_1+p-k+1}{k-1}\\
	p &\geq n_2 -n_1  - \frac{n_2-k+1}{k}.
	\end{align*}
	Hence the minimum value of $p$ possible is $p^{*}=  \lceil n_2 -n_1 - \frac{n_2-k+1}{k} \rceil$. 
	\item we need to find $$n^{+}=\min{\{n' \in [\max{\{n_1,n_2-6\}},n_2-1]\}}$$ such that $C^{*}_{n^{+}} =\max_{n'\in [\max{\{n_1,n_2-6\}},n_2-1]} C^{*}_{n'}$. 
	\item we also need to find $$n_{min} = \min{\{n'' \in [\max{\{n_1,n_2-6\}},n_2-1]\}} $$ such that $n_2 \le n'' + C^{*}_{n''} $.

\end{enumerate} 

Our objective is to maximize $G_1$ in Theorem \ref{thm: main results}. Considering point $1$ alone, which is mentioned above, the maximum $G_1$ possible is $G^{1}_1=\min{\{n_2-(n_1+p^{*}),C_{n_1+p^{*}}\}}$. If we consider the point $2$, the maximum $G_1$ possible is $G^{2}_1=\min{\{n_2-n^{+},C^{*}_{n^{+}}\}}$. Similarly, if we consider the point $3$, the maximum $G_1$ possible is $G^{3}_1 =n_2-n_{min}$. So summarizing all the three points, the maximum $G_1$ is $\max{\{G^{1}_1,G^{2}_1,G^{3}_1\}}$. If $G_1=G^{1}_1$, let $n'_1= n_1+p^{*}$, else if $G_1=G^{2}_1$, let $n'_1= n^{+}$, else, let $n'_1 =n_{min}$. For $n_2 \ge 3(k-1)$, we can use Construction \ref{constr:cyclic_gen_1} where we replace $n_1$ with $n'_1$ servers. Initially, we launch $n_1$ of $n'_1$ servers. After one of those servers finish their tasks, the remaining $n_2-n_1$ servers are launched.
\end{constr}

\begin{note}
In this case, for general $n_2$, the amount of computation per server required is $\frac{n_2 -k+1 -G_1 }{n_2-G_1}$, where $G_1=	\max\{ 	\min\{n_2-(n_1+p^{*}),\lfloor \frac{n_1+p^{*}-k+1}{k-1} \rfloor\},$ $ \min\{n_2-n^{+},C^{*}_{n^{+}}\}, n_2-n_{min} \}, $ $p^{*}=  \lceil n_2 -n_1 - \frac{n_2-k+1}{k} \rceil$, $n^{+}=\min{\{n' \in [\max{\{n_1,n_2-6\}},n_2-1]\}}$ such that $C^{*}_{n^{+}} =\max_{n'\in [\max{\{n_1,n_2-6\}},n_2-1]} C^{*}_{n'}$, $n_{min} = \min{\{n'' \in [\max{\{n_1,n_2-6\}},n_2-1]\}} $ such that $n_2< n'' + C^{*}_{n''} $. This proves Theorem \ref{thm: main results} for $n_1,n_2 \geq 3(k-1)$.
\end{note}
\subsection{Tiered Gradient Codes for $2(k-1) < n_1  < 3(k-1), n_2 \geq 3(k-1)$ }
\label{sec: c1 mid range}
In this subsection, we provide tiered gradient codes for the case where $2(k-1) < n_1  < 3(k-1),  n_2 \geq 3(k-1)$. 
The construction is in similar lines to that of Construction \ref{const: gen_c1_n1>3(k-1)}, except that,

\begin{itemize}
\item for the point $1$ mentioned in Construction \ref{const: gen_c1_n1>3(k-1)}, the 'p' should satisfy one more condition, i.e., $n_1 +p \ge 3(k-1)$. So the minimum value of $p$ possible is $p^{*}=\max{\{3(k-1)-n_1,  \lceil n_2 -n_1 - \frac{n_2-k+1}{k} \rceil\}}$. 
\item we have to replace $n_1$ with $3(k-1)$ in point $2$ of Construction \ref{const: gen_c1_n1>3(k-1)}, i.e., we need to find $n^{+}=\min{\{n' \in [\max{\{3(k-1),n_2-6\}},n_2-1]\}}$ such that $C^{*}_{n^{+}} =\max_{n'\in [\max{\{3(k-1),n_2-6\}},n_2-1]} C^{*}_{n'}$. 
\item for point $3$ of Construction \ref{const: gen_c1_n1>3(k-1)} also, we need to replace $n_1$ with $3(k-1)$, i.e., we also need to find $n_{min} = \min{\{n'' \in [\max{\{3(k-1),n_2-6\}},n_2-1]\}} $ such that $n_2 \le n'' + C^{*}_{n''} $.

\end{itemize}

Our aim is to maximize $G_2$ in Theorem \ref{thm: main results}. So considering all the points mentioned above, the maximum $G_2$ is $\max{\{G^{1}_2,G^{2}_2,G^{3}_2\}}$, where $G^{1}_2=\min{\{n_2-(n_1+p^{*}),C_{n_1+p^{*}}\}}, G^{2}_2=\min{\{n_2-n^{+},C^{*}_{n^{+}}\}}$ and $G^{3}_2 =n_2-n_{min}$. 
If $G_2=G^{1}_2$, let $n'_1= n_1+p^{*}$, else if $G_2=G^{2}_2$, let $n'_1= n^{+}$, else, let $n'_1 =n_{min}$. For $n_2 \ge 3(k-1)$, we can use Construction \ref{constr:cyclic_gen_1} where we replace $n_1$ with $n'_1$ servers. Initially, we launch $n_1$ of $n'_1$ servers. After one of those servers finish their tasks, the remaining $n_2-n_1$ servers are launched.
%
%

\begin{note}
In this case, for general $n_2$, the amount of computation per server required is $\frac{n_2 -k+1 -G_2 }{n_2-G_2}$, where $G_2=	max\{ min\{	n_2-(n_1+p^{*}),\lfloor \frac{n_1+p^{*}-k+1}{k-1} \rfloor\},$ $ min\{n_2-n^{+},C^{*}_{n^{+}}\}, n_2-n_{min} \},\\ p^{*}=max{\{3(k-1)-n_1,  \lceil n_2 -n_1 - \frac{n_2-k+1}{k} \rceil\}}, n^{+}=min{\{n' \in [max{\{3(k-1),n_2-6\}},n_2-1]\}}$ such that $C^{*}_{n^{+}} =max_{n'\in [max{\{n_1,n_2-6\}},n_2-1]} C^{*}_{n'}$, $n_{min} = min{\{n'' \in [max{\{3(k-1),n_2-6\}},n_2-1]\}} $ such that $n_2 \le n'' + C^{*}_{n''} $. This proves Theorem \ref{thm: main results} for $2(k-1) < n_1  < 3(k-1), n_2 \geq 3(k-1),n_2 \geq 3(k-1)$.
\end{note}

\section{Tiered Gradient Codes for $c>1$}
\label{sec: c>1}
In this section, we deal with the case where initially we launch the first $n_1$ servers and wait for $c>1$ servers to complete their tasks. After that the remaining $n_2-n_1$ servers are launched.

This section is organized as follows. Initially, we discuss about tiered gradient codes for $k \le n_1 \leq 2(k-1)$ and $n_2 > 2(k-1)$. Then we move on to $n_1,n_2 \geq 2(k-1) + (k-c)$.  Towards the end, we provide a discussion on codes for $2(k-1) < n_1 < 2(k-1) + (k-c)$ and $n_2 \ge 2(k-1) + (k-c)$.

For $k \le n_1 \leq 2(k-1)$ and $n_2 > 2(k-1)$, we use Construction \ref{constr:frac1}, where we wait for $c$ servers to complete their tasks instead of one server.

For any specific $n_1 > 2(k-1)$, we use the unique cyclic repetition gradient code for the first $n_1$ servers.  We need to show that the Condition \ref{eq:supp_cond_tgc} holds for the codes under consideration. Let $M=\{i_1,i_2,...,i_c\}$. 	If $\mathcal{T}_2 = \phi$, Condition \ref{eq:supp_cond_tgc} follows from the support structure of cyclic repetition code. If $|\mathcal{T}_2| = 1$ and some subset of $M$ is included in $ \mathcal{T}_1$, the support of the union is $[n_1]$ and hence Condition \ref{eq:supp_cond_tgc} is satisfied. Consider the case where $|\mathcal{T}_2| = 1, M \notin \mathcal{T}_1$, $|\mathcal{T}_1| = k-c-1$ and $\mathcal{T}_1$ is such that
\begin{equation*}
|\cup_{i \in \mathcal{T}_1} L_i | = n_1 - k + (k-c-1) = n_1-c-1.
\end{equation*}
Based on the cyclic support structure of the $F$ matrix, Condition \ref{eq:supp_cond_tgc} is true whenever $(k-c-1)$ consecutive rows (modulo $n_1$) are picked. Hence, the $c+1$ coordinates which are not included in the union are consecutive. So, each row in the $B_M$ matrix should be designed in such a way that at least one of every $c+1$ consecutive coordinates should be non zero. 

Now, we will construct codes where at least one of every two consecutive coordinates is non zero in the $B_M$ matrix. We consider $Q=n_1$. Let $C'_{n_1} =\lfloor \frac{n_1-k+c}{k-1} \rfloor$. Initially, we launch $n_1$ servers. Let $\bold{f}_i$ represent the $i^{th}$ row of the $F$ matrix and $L_i$ represent the support of $\bold{f}_i$, where $i \in [n_1]$. Let $\bold{b}_i$ represent the $i^{th}$ row of the $B_M$ matrix and $Z_i$ represent the support of $\bold{b}_i$, where $i \in [C'_{n_1}]$. Let the columns of the $F$ and $B_M$ matrices be indexed by $[0,n_1-1]$. Let $\{i_1,i_2, \ldots , i_c\}$ be the $c$ servers who complete their tasks first. Then, $\{L_{i_1},L_{i_2}, \ldots , L_{i_c}\}$ is the support of $\{\bold{f}_{i_1},\bold{f}_{i_2}, \ldots , \bold{f}_{i_c}\}$ respectively.

The code construction for $n_1 \geq 2(k-1) + (k-c)$ and $n_2 = n_1 + C'_{n_1}$ is as follows. 

\begin{constr}
	\label{constr:cyclic_gen_c1}
	($n_1 \geq 2(k-1) +(k-c), n_2 = n_1 + C'_{n_1}$)
	The support structure of the matrix $F$ is as follows:
	$
	supp(\bold{f_i}) = [i-1, i+(n_1-k-1)] \mod n_1.
	$
	For some $t \in [0,n_1-1]$, let 
	\begin{align*}
	n_1 - |\cup_{i \in \{i_1,i_2, \ldots , i_c\}}L_i| &=g,\\
	[0,n_1-1] \setminus \cup_{i \in \{i_1,i_2, \ldots , i_c\}}L_i &=[t,t+g-1], \\
	l&=(n_1-k+c)-(k-1)C'_{n_1}.
	\end{align*}
	
	The $l+k-c$ elements of $Z_j$, for each $j \in [C'_{n_1}]$, are given by $
	[t-l,t+k+c-1] \subset Z_j, \ \ j=[C'_{n_1}].  $
	
	If $C'_{n_1} = 1$,  pick the remaining coordinates so that consecutive coordinate property is satisfied, i.e, at least one coordinate from every possible pair of two consecutive coordinates are picked up. Else if $C'_{n_1} > 1$, do the following.  
	Let $B^{*}_{M}$ be the matrix obtained by shifting all the columns (say, $y$ number of shifts done to each column towards right) in the $B_M$ matrix in such a way that the columns -$[t-l, t+k-c-1\}$ of the $B_M$ matrix become the last $l+k-c$ columns in $B^{*}_{M}$.

	\[
	B^{*}_{M}=
	\begin{bmatrix}
	B_{M_1} & B_{M_2} & \ldots &  B_{M_{k-1} } & B_{M'}
	\end{bmatrix}
	\]
	Let $Z^{*}_{i}$ represent the support of the $i^{th}$ row of the $B^{*}_{M}$ matrix.  $B_{M'}$ is the matrix obtained by taking the  last $l+k-c$ columns of the $B^{*}_{M}$ matrix. All the entries in the  $B_{M'}$ matrix are non zero, i.e.,
	\begin{equation*}
	[n_1-(l+k-c), n_1 -1]\subset Z^{*}_j,j \in [C'_{n_1}].
	\end{equation*} 
	
	$B_{M_1}$ constitutes of the first $C'_{n_1}$ columns of the $B^{*}_{M}$ matrix, $B^{*}_{M_2}$ constitutes of the next $C'_{n_1}$ columns and so on.
	Each $B_{M_j}$ is a $C'_{n_1} \times C'_{n_1}$ matrix which is obtained by taking distinct and consecutive $C'_{n_1}$ columns from the $B^{*}_M$ matrix sequentially.  
	
	The support of the $i^{th}$ row of each matrix $B_{M_j}$, where $j \in [k-1]$, is of the form $[i-1, i+C'_{n_1}-3]$ mod $C'_{n_1}$. The support 
	structure of the $B_M$ matrix is same as that of the $B^{*}_M$ matrix with each column of the $B^{*}_M$ matrix shifted towards left by $y$. 	The construction of the $B_M$ and $F$ matrices using the above support structures is same as in Construction \ref{constr:cyclic_gen_1}.

\end{constr}
\begin{proof}

	From Lemma \ref{lem:supp_cond}, we have $
	[0,n_1-1] \setminus \cup_{i \in \{i_1,i_2, \ldots , i_c\}}L_i \subset Z_j, \ \ j=[C'_{n_1}].  $
	Thus $g$ coordinates are included in each $Z_j$, for each $j \in [C'_{n_1}]$. We have to add $|Z_j|-g = n_1-k+1-g$ more coordinates to $Z_j$ from $\cup_{i \in \{i_1,i_2, \ldots , i_c\}}L_i$. That is, we need to pick $n_1-k+1-g$ from $n_1-g$ locations available. Hence for the consecutive coordinate property to be satisfied, $n_1-k+1-g \geq \lfloor \frac{n_1-g}{2} \rfloor$, i.e., $n_1 \geq 2(k-1) + g$. The maximum value that $g$ can take is $k-c$, which is basically when all the $c$ servers who finish first are consecutive ones. Considering the worst case scenario, the consecutive coordinate property is satisfied when $n_1 \geq 2(k-1) + (k-c)$, which is our range of $n_1$ for which the code is constructed.

	The procedure to design the support of each row of the $B_M$ matrix is as follows. It is required that $|Z_{j}| = n_1-k+1$, for each $j \in [C'_{n_1}]$. No two servers among the first $n_1$ servers can have disjoint data set. It comes from the fact that $
	2(n_1-k+1) =2n_1-2(k-1)>2n_1-n_1=n_1.$
	The inequality in the second step is satisfied since $n_1 \geq 2(k-1) + (k-c)$. Hence  the set $[0,n_1-1] \setminus \cup_{i \in \{i_1,i_2, \ldots , i_c\}}L_i $ contains consecutive $g$ coordinates.
	
	Pick any consecutive $k-c$ coordinates from $[0, n_1-1]$ which includes the above $g$ coordinates. Let it be $\{t,t+1, \ldots, t+k-c-1\}$. Let $
	[t,t+k-c-1]  \subset Z_j, \ \ j=[C'_{n_1}].
	$.
	Thus $k-c$ coordinates are included in each $Z_j,j=[C'_{n_1}]$. Rest of the coordinates of $Z_j$, are picked to satisfy the consecutive coordinate property. And also, these are picked so that $|Z_j \cup Z_i| = n_1$, for any $j,i \in [C'_{n_1}]$.	
	The $l$ coordinates - $[t-l, t-1 ]$ are also included in $Z_j$. Thus, totally, $l+k-c$ coordinates are included in each $Z_j$. We have to add $|Z_j|-(l+k-c) = n_1-2k+1+c-l$ more coordinates to $Z_j$ from $ [0, n_1-1]\setminus [t-l, t+k-c-1]$. That is, we need to pick $n_1-2k+1+c-l$ from $n_1-k+c-l$ locations available.
	
	If $C'_{n_1} = 1$,  we pick the remaining coordinates so that consecutive coordinate property is satisfied. Else if $C'_{n_1} > 1$, we do the following.	
	Let $B^{*}_{M}$ be the matrix obtained by shifting all the columns in the $B_M$ matrix in such a way that the columns -$[t-l, t+k-c-1\}$ in the $B_M$ matrix become the last $l+k-c$ columns in  the $B^{*}_{M}$ matrix. Let $y$ be the number of shift done to each column of the $B_{M}$ matrix towards right to obtain the $B^{*}_{M}$ matrix. Let $Z^{*}_{i}$ represent the support of the  $i^{th}$ row of the $B^{*}_{M}$ matrix.
	
	Let $B_{M'}$ be the matrix obtained by taking the last $l+k-c$ columns of the $B^{*}_{M}$ matrix. All the entries in the $B_{M'}$ matrix are non zero, i.e., $[n_1-(l+k-c), n_1 -1]\subset Z^{*}_j,j \in [C'_{n_1}]$. 
	$B_{M_1}$ constitutes of the first $C'_{n_1}$ columns of the $B^{*}_{M}$ matrix, $B^{*}_{M_2}$ constitutes of the next $C'_{n_1}$ columns and so on.
	Hence, each $B_{M_j}$ is a $C'_{n_1} \times C'_{n_1}$ matrix which is obtained by taking distinct and consecutive $C'_{n_1}$ columns from the $B^{*}_M$ matrix.  	
	We have the support structure of the $B_{M'}$ matrix. The support structure for the remaining coordinates of the $B^{*}_M$ matrix, i.e., $n_1-2k+c+1-l$ more coordinates to be added to $Z^{*}_j$, is obtained from the design of the support structure corresponding to the matrices $B_{M_j}, j \in [k-1]$.
	
	The support of the $i^{th}$ row of each matrix $B_{M_j}$, where $j \in [k-1]$, is of the form $[i-1, i+C'_{n_1}-3]$ mod $C'_{n_1}$. The cardinality of the support of each row of the $B_{M_j}$ matrix is $C'_{n_1}-1$, i.e, there is exactly one zero in each row of the $B_{M_j}$ matrix at disjoint locations. Hence, the number of zeros in each row of the $B^{*}_{M}$ matrix is exactly $k-1$, which is exactly what we needed. 
	
	The cardinality of the support of union of any two rows of  the $B_{M_j}$ matrix is $C'_{n_1}$. Hence if we take union of support of any two rows in the $B^{*}_{M}$ matrix, then it has cardinality $n_1$. That is, $|Z_{r} \cup Z_{s}| = n_1$, for any $r,s \in [C'_{n_1}]$. Since the $B^{*}_M$ matrix is obtained by column shift of the $B_{M}$ matrix, the above property holds for the $B_M$ matrix also. Hence the support structure of the  $B_M $ matrix satisfies all the required conditions.
	
\end{proof}

\begin{example} {\label{exmp: 2}}
	Let $n_1=9,k=4,c=2$ and $n_2=11.$ We split data into $9$ partitions -$\{x_0,x_1,\ldots , x_8\}$. The server $i$ is assigned partitions $\{x_j, j \in [i-1, i+4]\}$. Each server computes the gradients on their respective data. Suppose server $1$ and $3$ finish their tasks first when $n_1$ servers are launched. After that the remaining $n_2-n_1=2$ servers are launched. Server $1$ or $3$ do not have $\{x_8\}$ as their contents. Here, $k-c=2$ and $l=1$. Hence we have to include  $\{x_7,x_8,x_0\}$ in the content of the two added servers.  The first column and the last two columns of the $B_M$ matrix is filled with non zero entries.  Shift each column of the $B_M$ matrix by $8$ units towards right to obtain the $B_M^{*}$ matrix, $
	B_{M}^{*}=
	\begin{bmatrix}
	B_{M_1} & B_{M_2}&B_{M_3} & B_{M'}
	\end{bmatrix}
	$. 
	The $B_{M'}$ matrix is obtained by taking the last three columns of the $B^{*}_M$ matrix. Hence it is a $2 \times 3 $ matrix. $B_{M_1},B_{M_2}$ and $B_{M_3}$ are $3 \times 3$ matrices. $B_{M_1}$ is the submatrix formed by the first columns of the $B^{*}_{M}$ matrix, $B_{M_2}$ is formed by the next three columns of the $B^{*}_{M}$ matrix  and $B_{M_3}$ by the next three columns. The support of the $i^{th}$ row of each of the matrices  $B_{M_1},B_{M_2}$ and $B_{M_3}$ is of the form $\{i-1\}$ mod $2$. Hence the structures of the $B_{M_1},B_{M_2},B_{M_3} $ and $B^{*}_M$ matrices are of the form
	$
	B_{M_1}= B_{M_2}=B_{M_3}=
	\begin{bmatrix}
	0 & *   \\
	* & 0   \\
	\end{bmatrix}$ and $
	B^{*}_{M}=
	\begin{bmatrix}
	0 & * & 0 & * & 0 & * & * & * & *  \\
	* & 0 & * & 0 & * & 0 & * & *  & *\\
	\end{bmatrix}
	$.
	The entries $*$ represent non zero values. Shift each column, in the  $B^{*}_{M}$ matrix, $8$ units towards left to obtain the support structure of the $B_M$ matrix, $
	B_{M}=
	\begin{bmatrix}
	*&0 & * & 0 & * & 0 & * & * & *   \\
	*&* & 0 & * & 0 & * & 0 & * & * \\
	\end{bmatrix}
	$.  The content of the two added servers are $\{x_0,x_2,x_4,x_6,x_7,x_8\}$ and  $\{x_0,x_1,x_3,x_5,x_7,x_8\}$ respectively. Out of the $8$ servers which haven't finished the job earlier and the two added servers, any three servers can give the sum of the gradients along with server $1$ and $3$. Each server does $6/9$ computations compared to $8/11$ required for the $(n_2,k)$ gradient code.
	
\end{example}

For $n_1 \geq 2(k-1)+(k-c),  n_2 < n_1 + C_{n_1}$, we take the support structure of any $n_2-n_1$ rows of the $B_M$ matrix constructed using Construction \ref{constr:cyclic_gen_c1} ($ n_2 = n_1 +C'_{n_1}$) to generate the support structure for the $B_M$ matrix in this case. The support structure of the matrix $F$, the construction of the $B_M$ and $F$ matrices using the above support structures is same as in Construction \ref{constr:cyclic_gen_1}.


\begin{example}
	Consider Example \ref{exmp: 2} with $n_1=9,k=4,c=2$. Consider $n_2=10$. Here, $C'_{n_1}=2$. Hence, $n_2-n_1 = 1<2$. 
	The setting is same as in Example \ref{exmp: 2}. The only difference is that $ n_2-n_1 < C'_{n_1}$. 	
	Hence we can use any one row of the $B_M$ matrix from Example \ref{exmp: 2} to generate the $B_M$ matrix for this case. Let us take the first  row. Hence, $
	B_{M}=
	\begin{bmatrix}
	* & 0 & * & * & 0 & * & * & * & *\\
	\end{bmatrix}
	$,
	where the symbol $*$ represent non zeros entries. Hence the content of the one added server is  $\{x_0,x_2,x_3,x_5,x_6,x_7,x_8\}$. Out of the $8$ servers which haven't finished the job earlier and the one added servers, any three servers can give the sum of the gradients along with server $1$. Each server does $7/9$ computations compared to $8/10$ required for the $(n_2,k)$ gradient code.
\end{example}

For $n_1 \geq 2(k-1) + (k-c), n_2  > n_1 +C'_{n_1}$, we need to find a $p$ such that the following condition is satisfied, i.e.,
\begin{align*}
n_2 &\leq n_1 +p + C'_{n_1+p}\\
& =n_1+p+\frac{n_1+p-k+c}{k-1}\\
p &\geq n_2 -n_1 - \frac{n_2-k+c}{k}.
\end{align*}
Hence the minimum value of $p$ possible is $p^{*}=\lceil n_2 -n_1 - \frac{n_2-k+c}{k} \rceil$. So, for $n_2 > n_1 + C'_{n_1}$, we can use Construction \ref{constr:cyclic_gen_c1} where we replace $n_1$ with $n_1+p^{*}$ servers. Initially, we launch $n_1$ of $n_1 + p^{*}$ servers. After $c$ of those servers finish their tasks, the remaining $n_2-n_1$ servers are lauched.


\begin{example}
	Consider $n_1=9,k=4,c=2$ and $n_2=12.$ Here, $C'_{n_1}=2$. Hence $n_2 > n_1 + C_{n_1}$. Here, $p^{*} = 1 $ and $C_{n_1+p^{*}}=2$. So, we split data into $10$ partitions -$\{x_0,x_1,\ldots , x_9\}$. The server $i$ is assigned data $\{x_j, j \in [i-1,i+5]$. Each server computes the gradients on their respective data. 
	
	Initially the first $n_1$ servers are launched. Suppose server $2$ and $4$ finish their tasks first. After that all the remaining servers are launched. Since server $2$ or $4$ doesn't have $\{x_1\}$ as their content, $k-c=2$ and $l=n_1+p^{*} - k+1-(k-1) C_{n_1+p^{*}}=2$, we have to include  $\{x_0,x_9,x_1,x_2\}$ in the content of server $11$ and $12$. The last three and the first column of the $B_M$ matrix are filled with non zero entries. Shift each column of the $B_M$ matrix by $9$ units towards right to obtain $B_M^{*}$ matrix.   
	$
	B_{M}^{*}=
	\begin{bmatrix}
	B_{M_1} & B_{M_2}&B_{M_3} & B_{M'}
	\end{bmatrix}
	$. 
	$B_{M'}$ is a $2 \times 3 $ matrix obtained by taking the last four columns of $B^{*}_M$. $B_{M_1},B_{M_2}$ and $B_{M_3}$ are $2 \times 2$ matrices. $B_{M_1}$ is the submatrix formed by the first two columns of the $B^{*}_{M}$ matrix, $B_{M_2}$ is formed by the next two columns of the $B^{*}_{M}$ matrix and $B_{M_3}$ by the next two columns. The support of the $i^{th}$ row of each of the matrices  $B_{M_1},B_{M_2}$ and $B_{M_3}$ is of the form $\{i-1\}$ mod $2$. Hence the structures of the $B_{M_1},B_{M_2},B_{M_3} $ and $B^{*}_M$ matrices are of the form
	$
	B_{M_1}= B_{M_2}=B_{M_3}=
	\begin{bmatrix}
	0 & *   \\
	* & 0   \\
	\end{bmatrix}
	$ and $
	B^{*}_{M}=
	\begin{bmatrix}
	0 & * & 0 & * & 0 & * & * & * & * & * \\
	* & 0 & * & 0 & * & 0 & * & *  & * & *\\
	\end{bmatrix}
	$.
	The entries $*$ represent non zero values. Shift each column, in the  $B^{*}_{M}$ matrix, $9$ units towards left to obtain the support structure of the $B_M$ matrix, $
	B_{M}=
	\begin{bmatrix}
	*&*&*&0 & * & 0 & * & 0 & * & *    \\
	*&*&*&* & 0 & * & 0 & * & 0 & *  \\
	\end{bmatrix}
	$.
	The content of the two added servers are $\{x_0,x_1,x_2,x_4,x_6,x_8,x_9\}$ and  $\{x_0,x_1,x_2,x_3,x_5,x_7,x_9\}$ respectively. Out of the $9$ servers which haven't finished the job earlier and the two added servers, any three servers can give the sum of the gradients along with servers $2$ and $4$. Each server does $7/10$ computations compared to $9/12$ required for the $(n_2,k)$ gradient code.
\end{example}

\begin{thm} \label{thm:cyclic_C_n1}
	The code given in Construction \ref{constr:cyclic_gen_c1} is a tiered gradient code where $n_1,n_2 \geq 2(k-1) +(k-c)$.
\end{thm}
\begin{proof}
	
	We need to prove that Condition \ref{eq:supp_cond_tgc} is satisfied by the code in Construction \ref{constr:cyclic_gen_c1}.  Let $M=\{i_1,i_2,...,i_c\}$. 	If $\mathcal{T}_2 = \phi$, Condition \ref{eq:supp_cond_tgc} follows from the support structure of the cyclic repetition code. If $|\mathcal{T}_2| = 1$ and some subset of $M$ is included in $ \mathcal{T}_1$, then the support of the union is $[n_1]$ and hence Condition \ref{eq:supp_cond_tgc} is satisfied. For the case of $|\mathcal{T}_2| \geq 2$, since $Z_i$ and $Z_j$ are chosen such that $Z_i \cup Z_j = [n_1]$, for any $i,j \in [C'_{n_1}]$, we have that Condition \ref{eq:supp_cond_tgc} is trivially satisfied.
\end{proof}

\begin{constr}\label{const: gen_c>1_n1>3(k-1)}
	\textit(General $n_1,n_2 \ge 2(k-1)+(k-c)$).
	If $n_1,n_2 \ge 2(k-1)+(k-c),$
	
	\begin{enumerate}
		\item  we need to find a $p$ such that the following condition is satisfied, i.e.,
		\begin{align*}
		n_2 &\leq n_1 +p + C'_{n_1+p}\\
		& =n_1+p+\frac{n_1+p-k+c}{k-1}\\
		p &\geq n_2 -n_1 - \frac{n_2-k+c}{k}.
		\end{align*}

	\end{enumerate} 
	Hence the minimum value of $p$ possible is $p^{*}=\lceil n_2 -n_1 - \frac{n_2-k+c}{k} \rceil$. 	Our objective is to maximize $G_3$ in Theorem \ref{thm: main results}.  So considering the above mentioned point, the maximum $G_3$ possible is $\min{\{n_2-(n_1+p^{*}),C'_{n_1+p^{*}}\}}$. So, for $n_2,  n_1 \ge 2(k-1)+(k-c)$, we can use Construction \ref{constr:cyclic_gen_c1} where we replace $n_1$ with $n_1+p^{*}$ servers. Initially, we launch $n_1$ of $n_1 + p^{*}$ servers. After $c$ of those servers finish their tasks, the remaining $n_2-n_1$ servers are launched. 
\end{constr}

\begin{note}
	For $n_1,n_2 \geq 2(k-1) + (k-c)$, the computation per server required is proportional to $\frac{n_2 -k+1-G_3 }{n_2-G_3}$, where $G_3=\min{\{n_2-(n_1+p^{*}),\lfloor \frac{n_1+p^{*}-k+c}{k-1} \rfloor \}},$ $p^{*}=\lceil n_2 -n_1 - \frac{n_2-k+c}{k} \rceil$.  This proves Theorem \ref{thm: main results} for general $c$ and $n_1 \geq 2(k-1) + (k-c)$.
\end{note}

For $2(k-1) < n_1 < 2(k-1) + (k-c)$ and $n_1 \geq 2(k-1) + (k-c)$, the construction is in similar lines to that of Construction \ref{const: gen_c>1_n1>3(k-1)}, except that,

\begin{itemize}
	\item for the point $1$ mentioned in Construction \ref{const: gen_c>1_n1>3(k-1)}, the 'p' should satisfy one more condition, i.e., $n_1 +p \ge 2(k-1)+(k-c)$. So the minimum value of $p$ possible is\\ $p^{*}=\max{\{(2(k-1)+(k-c))-n_1,  \lceil n_2 -n_1 - \frac{n_2-k+c}{k} \rceil\}}$. 
	
\end{itemize}

Our aim is to maximize $G_4$ in Theorem \ref{thm: main results}. So noting the point mentioned above, the maximum $G_4$ possible is $\min{\{n_2-(2(k-1)+(k-c)),C'_{n_1+p^{*}}\}}$. So, for $n_2 \ge 2(k-1)+(k-c)$, we can use Construction \ref{constr:cyclic_gen_c1}, where we replace $n_1$ with $n_1+p^{*}$ servers. Initially, we launch $n_1$ of $n_1 + p^{*}$ servers. After $c$ servers finish their tasks, the remaining $n_2-n_1$ servers are launched.

\begin{note}
	For $2(k-1) < n_1 < 2(k-1) + (k-c),n_2 \geq 2(k-1) + (k-c)$, the computation per server required is proportional to $\frac{n_2 -k+1-G_4 }{n_2-G_4}$, where $G_4=\min{\{n_2-(2(k-1)+(k-c)),\lfloor \frac{n_1+p^{*}-k+c}{k-1} \rfloor \}},$ $p^{*}=\max{\{(2(k-1)+(k-c))-n_1,  \lceil n_2 -n_1 - \frac{n_2-k+c}{k} \rceil\}}$. This proves Theorem \ref{thm: main results} for general $c$ and $n_1 \geq 2(k-1) + (k-c)$.
\end{note}

\begin{figure*}[hbt!]
	\begin{subfigure}[]{.31\textwidth}
		\includegraphics[width=\linewidth]{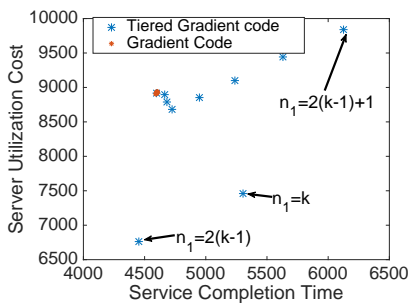}
		\caption{Task completion time distributed as SE1}
		\label{fig:SUC_Vs_SCT_k3_fixed_p}
	\end{subfigure}
	\hspace{.1in}
	\begin{subfigure}[]{.31\textwidth}
		\includegraphics[width=\linewidth]{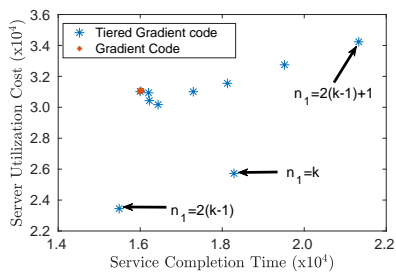}
		\caption{Task completion time distributed as SE2}
		\label{fig:SUC_Vs_SCT_k3_var_p}
	\end{subfigure}
	\hspace{.1in}
	\begin{subfigure}[]{.31\textwidth}
		\includegraphics[width=\linewidth]{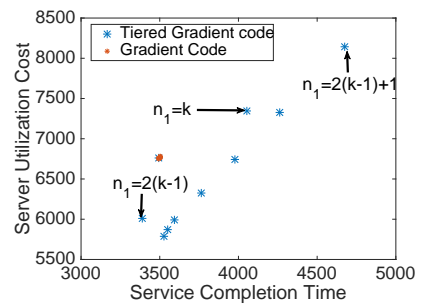}
		\caption{Task completion time distributed as Pa}
		\label{fig:SUC_Vs_SCT_k3_pareto}
	\end{subfigure}
	\caption{Server Utilization Cost as a function of Service Completion Time when we vary $n_1 \in [k,n_2]$ for $n_2=12$, $c=1$, and $k=3$.}\label{Figcase2}
\end{figure*}

\begin{figure*}[hbt!]
	\begin{subfigure}[]{.4\textwidth}
			\includegraphics[width=\linewidth]{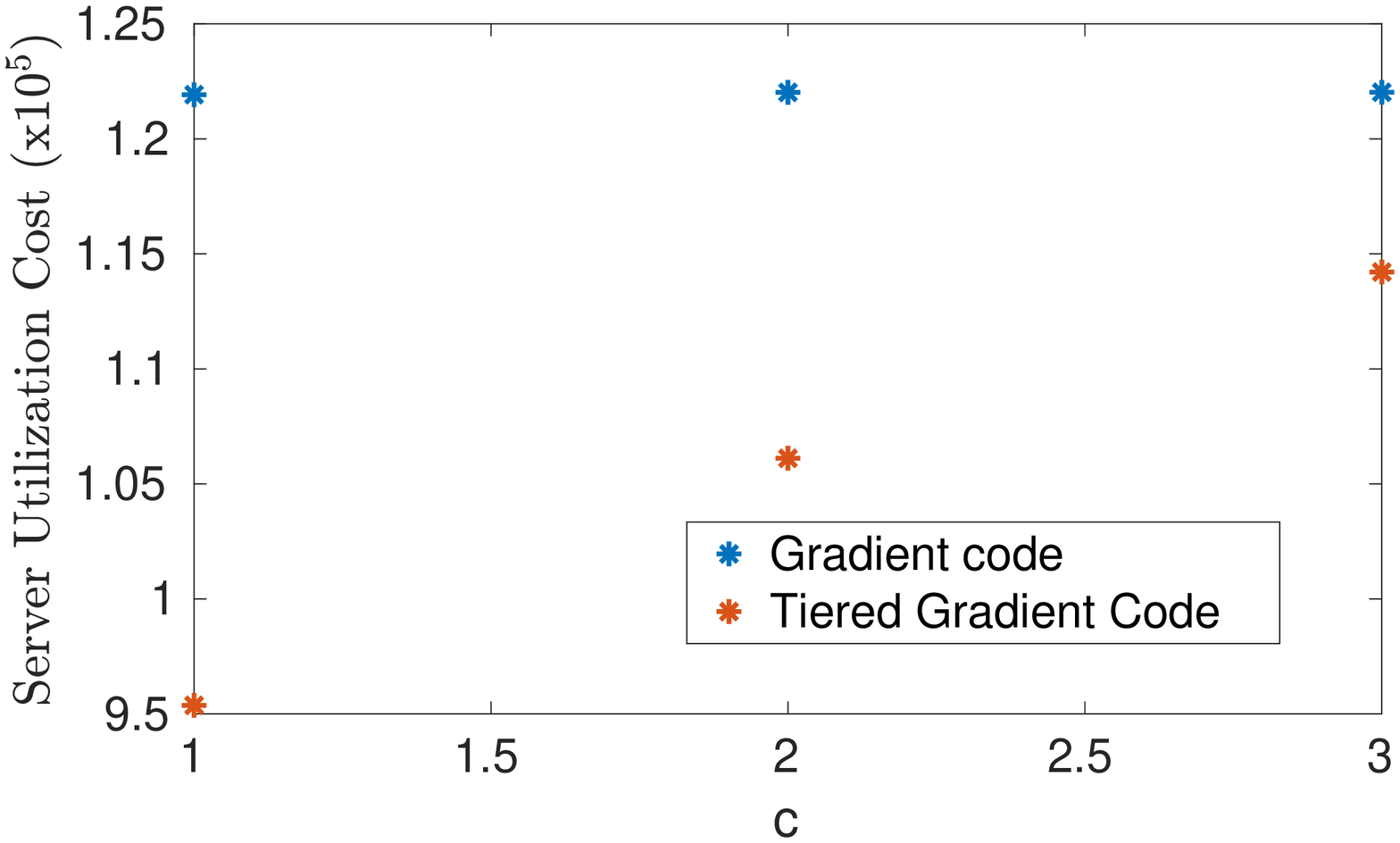}
		\caption{Server Utilization Cost as a function of $c$}
		\label{fig:SUC_Vs_c}
	\end{subfigure}
	\hspace{.1in}
	\begin{subfigure}[]{.4\textwidth}
			\includegraphics[width=\linewidth]{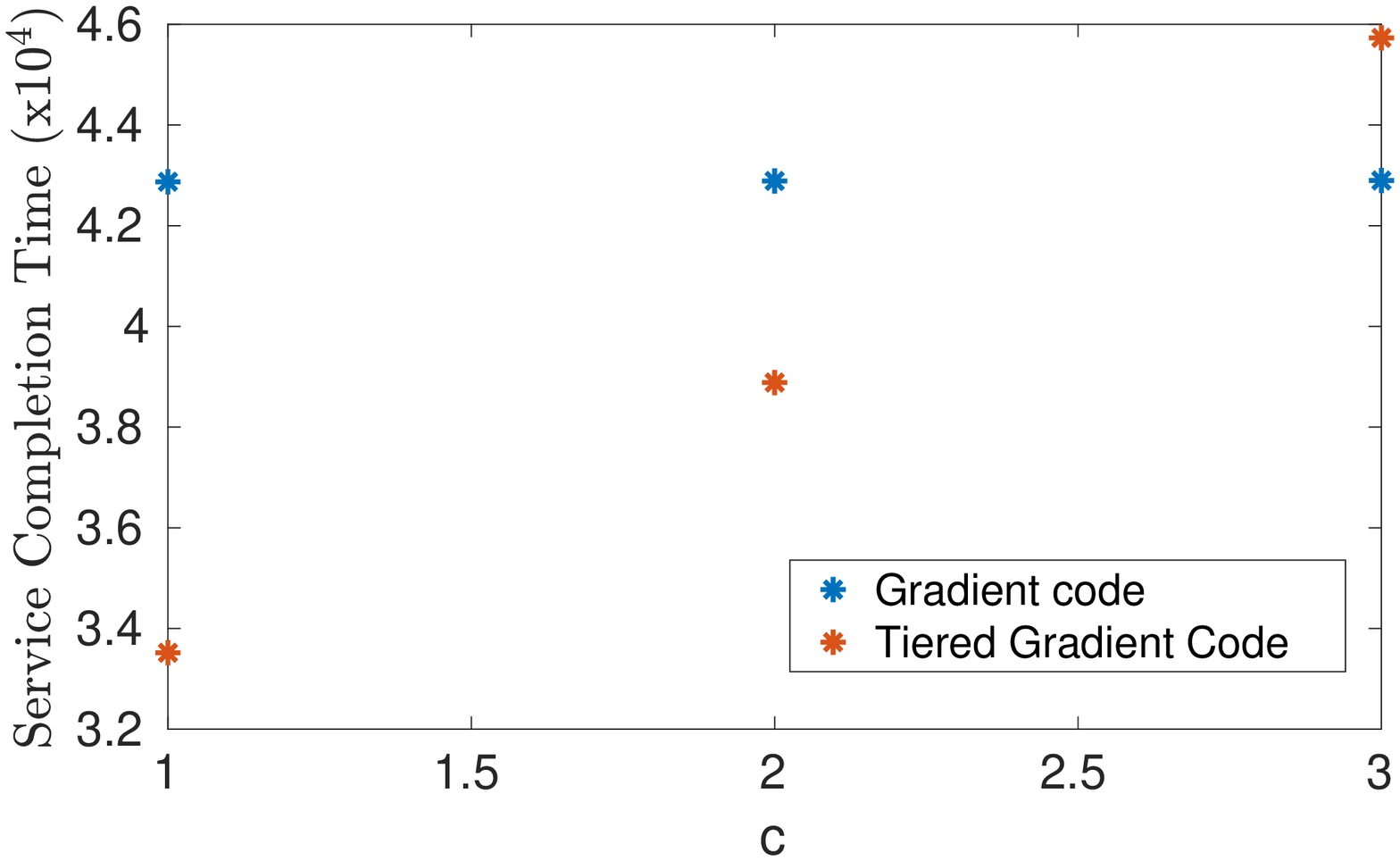}
		\caption{Service Completion Time as a function of $c$}
		\label{fig:SCT_Vs_c}
	\end{subfigure}
	\caption{Server Utilization Cost and Service Completion Time as a function of $c$ for $n_2=15$, $n_1=8$, and $k=5$. Task completion time is assumed to be distributed as SE2}\label{SE2varyc}
\end{figure*}

\section{Numerical Evaluations}\label{sec:nume}
In this section, we compare two metrics for the proposed tiered gradient codes to that for the gradient codes in \cite{tandon2017gradient}. The first metric is the service completion time, defined as the time taken for the $k$ tasks to complete. The second metric is the server utilization cost, which is the sum over all $n_2$ servers, the time during which each of the server is used till the job completes. Since both the metrics are random variables, dependent on the execution times of the tasks, we average the metrics over $10^4$ random trials to get mean results. 

Two distribution models are typically used to model the task execution times at the servers, both these distributions model the effect of stragglers in the job computation. The first is the shifted exponential distribution \cite{aggarwal2017taming,aktas2018straggler} which has probability distribution of task execution at each server as $\Pr(T>x) = e^{-\mu (x-d)^+}$ for all $x>0$, for the shift parameter $d$ and the mean parameter $\mu$. The mean parameter $\mu$ scales with the task size, and we assume that $\mu = 0.1$ times the computation per server requires as given in Theorem 1. The shift parameter happens from a combination of disk I/O and computation, and thus we consider two models for this. The first, called SE1, is where $d = 5$  times the computation per server requires as given in Theorem 1. The second, called SE2,  is where $d = 100$  and is independent the computation per server requires modeling more of the disk I/O rather than computation. The second distribution model that is considered is the Pareto distribution,  which has probability distribution of task execution at each server as $\Pr(T>x) = (\min(x,x_m)/x)^\alpha$ for all $x$, where $x_m$ is the scale parameter and $\alpha$ is the shape parameter. For our evaluations, we let $\alpha = 1.5$, and have $x_m = 1$ times the computation per server requires as given in Theorem 1. This distribution is label Pa. 

We first consider $n_2=15$, $c=1$, and $k=5$, and vary $n_1$ from $5$ to $15$. The tradeoff between the server utilization cost and service completion time for both the proposed codes and the codes in \cite{tandon2017gradient} are depicted in Fig. \ref{fig:SUC_Vs_SCT_k5_fixed_p}, \ref{fig:SUC_Vs_SCT_k5_var_p}, and \ref{fig:SUC_Vs_SCT_k5_pareto} for SE1, SE2, and Pa, respectively. In all three cases, the point with lowest service completion time and server utilization cost corresponds to $n_1=2(k-1)$.  Thus, the decrease in task size more than compensates the increase in expected completion time due to the delayed launching of $n_2-n_1$ tasks. The use of efficient tiered gradient codes decrease both the metrics significantly for $n_1=2(k-1)$ as compared to the gradient codes which corresponds to $n_1=n_2$. 
We also consider a different case - $n_2=12$, $c=1$, and $k=3$ and plot the trade off between the server utilization cost and service completion time in Fig.  \ref{fig:SUC_Vs_SCT_k3_fixed_p},  \ref{fig:SUC_Vs_SCT_k3_var_p}, and \ref{fig:SUC_Vs_SCT_k3_pareto} for SE1, SE2, and Pa, respectively, and achieve the same conclusions. We note that there is no monotonically relation with the parameters $n_1$ for the two metrics which are in part due to the code constructions having discrete changes.  The proposed codes helps choose parameters that  can help system designer trade off the two metrics more efficiently. In Fig.  \ref{fig:SUC_Vs_SCT_k5_fixed_p}, we see more than $25\%$ decrease in the both the metrics for tiered gradient codes at $n_1=2(k-1)$ as compared to the gradient codes thus showing that delayed relaunching is helpful and the code construction reduces the amount of computation efficiently.  


So far, we assumed $c=1$. We next consider the impact of general $c$. We let $n_2=15$, $n_1=8$, $k=5$  in Fig. \ref{SE2varyc}. 
We note that the service utilization cost decreases with $c$ since $n_2-n_1$ servers are not started at $t=0$, and wait till completion of $c$ servers. However, the service completion time increases with $c$ since the delayed starting of tasks lead to a delay in waiting for $k$ tasks to finish. However, for $c=1$ and $c=2$, both the metrics are significantly lower in the proposed approach as compared to the gradient codes. For $c=3$, the server utilization cost for the proposed codes is significantly lower for the proposed codes at an expense of the service completion time. Thus, both the metrics may need to be taken into account together for deciding the code parameters for the tiered gradient codes. The proposed codes gives additional degrees of flexibility in the design that can lead to significantly improved performance in the different metrics of the use of distributed servers, including the task per server, server utilization cost, and service completion time. 

\section{Conclusions}
\label{conclusion}
This paper provides a framework for tiered gradient codes where all redundant gradient computation servers are not launched at the same time. The framework assumes that when $c$ out of $n_1$ launched servers finish execution, $n_2-n_1$ additional servers can be launched, with a property that any $k$ of the servers can be used to compute the gradients. The framework allows for asynchronous launching of servers, and speculative execution by delayed launching of certain servers. Improvement in task computations per server is shown as compared to the case where all $n_2$ servers are launched without waiting for the results from $c$ out of $n_1$ servers.
%

\end{document}